\documentclass{article}
\usepackage{amsmath}

\usepackage{authblk}

\usepackage{graphicx}

\usepackage[margin=1in]{geometry}

\usepackage{amsthm}
\usepackage{amsfonts}
\usepackage{mdframed}

\newtheorem{theorem}{Theorem}
\newtheorem{lemma}[theorem]{Lemma}
\newtheorem{definition}[theorem]{Definition}
\newtheorem{corollary}[theorem]{Corollary}

\newtheorem{problem}[theorem]{Problem}
\newtheorem{proposition}[theorem]{Proposition}

\newcommand{\promprob}[1]{(#1_{\text{$\mathrm{YES}$}},#1_{\text{$\mathrm{NO}$}})}
\newcommand{\promprobY}[1]{#1_{\text{$\mathrm{YES}$}}}
\newcommand{\promprobN}[1]{#1_{\text{$\mathrm{NO}$}}}
\newcommand{\qpromprob}[1]{(\mathcal{#1}_{\text{$\mathrm{YES}$}},\mathcal{#1}_{\text{$\mathrm{NO}$}})}
\newcommand{\qpromprobY}[1]{\mathcal{#1}_{\text{$\mathrm{YES}$}}}
\newcommand{\qpromprobN}[1]{\mathcal{#1}_{\text{$\mathrm{NO}$}}}

\newcommand{\pr}{\text{$\mathrm{Pr}$}}

\newcommand{\cclass}[1]{\text{$\mathrm{#1}$}}

\title{Quantum State Isomorphism}

\author[1]{Joshua Lockhart}
\author[1,2,3]{Carlos E. Gonz\'alez-Guill\'en\thanks{This work was partially completed while the author was on a long term visit to Department of Computer Science, University College London.}}

\affil[1]{Department of Computer Science, University College London}
\affil[2]{Departamento de Matem\'atica Aplicada a la Ingenier\'ia Industrial, Universidad Polit\'ecnica de Madrid}
\affil[3]{Instituto de Matem\'atica Interdisciplinar, Universidad Complutense de Madrid}

\begin{document}
\maketitle

\begin{abstract}

We consider a problem we call \textsc{StateIsomorphism}: given two quantum states of $n$ qubits, can one be obtained from the other by rearranging the qubit subsystems? Our main goal is to study the complexity of this problem, which is a natural quantum generalisation of the problem \textsc{StringIsomorphism}. We show that \textsc{StateIsomorphism} is at least as hard as \textsc{GraphIsomorphism}, and show that these problems have a similar structure by presenting evidence to suggest that \textsc{StateIsomorphism} is an intermediate problem for \cclass{QCMA}. In particular, we show that the complement of the problem, \textsc{StateNonIsomorphism}, has a two message quantum interactive proof system, and that this proof system can be made statistical zero-knowledge. We consider also \textsc{StabilizerStateIsomorphism} (SSI) and \textsc{MixedStateIsomorphism} (MSI), showing that the complement of SSI has a quantum interactive proof system that uses classical communication only, and that MSI is QSZK-hard.

\end{abstract}

\section{Introduction and statement of results}
Ladner's theorem \cite{ladner} states that if $\cclass{P}\neq \cclass{NP}$ then there exists \emph{\cclass{NP}-intermediate problems}: \cclass{NP} problems that are neither \cclass{NP}-hard, nor in \cclass{P}. While of course the \cclass{P} \emph{vs.} \cclass{NP} problem is unresolved, the problem of testing if two graphs are isomorphic (\textsc{GraphIsomorphism}) has the characteristics of such an intermediate problem. \textsc{GraphIsomorphism} is trivially in NP, since isomorphism of two graphs can be certified by describing the permutation that maps one to the other, but as Boppana and H\aa stad show \cite{bh}, if it is NP-complete then the polynomial hierarchy collapses to the second level. Furthermore, while many instances of the problem are solvable efficiently in practice \cite{nauty}, it is still not known if there exists a polynomial time algorithm for the problem.

Recall that Quantum Merlin Arthur (\cclass{QMA}) is considered to be the quantum analogue of \cclass{NP}: the certificate is a quantum state, and the verifier has the ability to perform quantum computation. The class \cclass{QCMA} is defined in the same way but with certificates restricted to be classical bitstrings.
In this paper, we show that there are problems that exhibit similar hallmarks of being intermediate for \cclass{QCMA} \cite{an}. Succinctly: we formulate problems in \cclass{QCMA} that are not obviously in \cclass{BQP}, and which are unlikely to be \cclass{QCMA}-complete.
 
Babai's recent quasi-polynomial time algorithm for \textsc{GraphIsomorphism} \cite{babai} has revived a fruitful body of work that links the problem to algorithmic group theory \cite{babai2,groupgraph,luks,luks2}. This literature deals with a closely related problem called \textsc{StringIsomorphism}: given bitstrings $x,y\in\{0,1\}^n$ and a permutation group $G$, is there $\sigma \in G$ such that $\sigma(x)=y$ (where permutations act in the obvious way on the strings)? This problem has a number of similarities with \textsc{GraphIsomorphism}, and, as we show, can be recast in terms of quantum states.

We study what is arguably the most direct quantum generalisation of this problem, a problem we call \textsc{StateIsomorphism}. Such a generalisation is obtained by replacing the strings $x$ and $y$ by $n$-qubit pure states, and by considering the permutations in the group $G$ to act as ``reshufflings'' of the qubits.
The problem is obviously in \cclass{QCMA}: if there is a permutation mapping one state to the other then its permutation matrix acts as the certificate. Equality of two quantum states can be verified via an efficient quantum procedure known as the SWAP test \cite{swaptest}. Also, if there is an efficient quantum algorithm then the same can be used as an algorithm for \textsc{GraphIsomorphism}: as we shall see later, there exists a polynomial time many-one reduction from \textsc{GraphIsomorphism} to \textsc{StateIsomorphism}.

We first establish that in terms of interactive proof systems that solve the problem, \textsc{StateIsomorphism} has a number of similarities with its classical counterpart. A central part of the Boppana-H\aa stad collapse result is that \textsc{GraphIsomorphism} belongs in $\cclass{co-IP}(2)$: that is, that \textsc{GraphNonIsomorphism} has a two round interactive proof system. We show that \textsc{StateIsomorphism} is in $\cclass{co-QIP}(2)$: its complement has a two round \emph{quantum} interactive proof system. \textsc{GraphIsomorphism} also admits a statistical zero knowledge proof system, and indeed, we prove that \textsc{StateIsomorphism} has an honest verifier quantum statistical zero knowledge proof system. These results are summarised in the following theorem, where QSZK is the class of problems with (honest verifier) quantum statistical zero knowledge proof systems, defined by Watrous in \cite{qszk}. Note that since $\cclass{QIP}(2)\supseteq \cclass{QSZK} = \cclass{co-QSZK}$ (see \cite{qszk}), inclusion in $\cclass{co-QIP}(2)$ follows as a corollary.
\begin{theorem}
\label{theorem:SIQSZK}
 \textsc{StateIsomorphism} is in $\text{\emph{QSZK}}$.
\end{theorem}
A corollary of this theorem provides evidence to suggest that \textsc{StateIsomorphism} is not \cclass{QCMA}-complete. If it were, then every problem in \cclass{QCMA} would have an honest verifier quantum statistical zero knowledge proof system. 
Furthermore, this result is evidence against the problem being \cclass{NP}-hard: it is unlikely that $\cclass{NP}\subseteq \cclass{QSZK}$. 
\begin{corollary}
	 \label{corollary:qcma}
	 If \textsc{StateIsomorphism} is \emph{\cclass{QCMA}}-complete then \emph{$\cclass{QCMA}\subseteq \cclass{QSZK}$}.
\end{corollary}

In pursuit of stronger evidence against \cclass{QCMA}-hardness of \textsc{StateIsomorphism}, we consider a quantum polynomial hierarchy in the same vein as those considered by Gharibian and Kempe \cite{gk}, and Yamakami \cite{yamakami}. This hierarchy is defined in terms of quantum $\exists$ and $\forall$ complexity class operators like those of \cite{yamakami}, but from our definitions it is easy to verify that lower levels correspond to well known complexity classes. In particular, $\Sigma_0=\Pi_0=\cclass{BQP}$, and $\Sigma_1=\cclass{QCMA}$ or $\Sigma_1=\cclass{QMA}$ depending on whether we take the certificates to be classical or quantum (see Section \ref{section:aquantumpolynomialhierarchy}). Also, from the definition we provide, it is clear that the class $\text{cq-}\Sigma_2$ corresponds directly to the identically named class in \cite{gk}.

We prove the following, where $\cclass{QPH}=\cup_{i=1}^\infty \Sigma_i$, and $\cclass{QCAM}$ is the quantum generalisation of the class $\cclass{AM}$ where all communication between Arthur and Merlin is restricted to be classical \cite{mw}.
\begin{theorem}
	\label{theorem:collapse}
	Let $A$ be a promise problem in \emph{$\cclass{QCMA}\cap \text{co-}\cclass{QCAM}$}. If $A$ is \emph{$\cclass{QCMA}$}-complete, then \emph{$\cclass{QPH}\subseteq \Sigma_2$}.
\end{theorem}

While the relationship between the levels of this hierarchy and the levels of the classical hierarchy remains an open research question \cite{bqppolyhi}, the fact that the lower levels of this quantum hierarchy coincide with well known classes gives weight to collapse results of this kind.

 We draw attention to the fact that the collapse implication in Theorem \ref{theorem:collapse} is for the classical certificate classes \cclass{QCMA} and \cclass{QCAM}, rather than for the more well known \cclass{QMA} and \cclass{QAM} \cite{mw}. While the problems we consider are in \cclass{QCMA}, meaning that the current statement of the theorem is all we need, already we have an interesting open question: is there a similar collapse theorem that relates \cclass{QMA} and \cclass{QAM}? The proof of Theorem \ref{theorem:collapse} relies on the fact that $\cclass{QCMAM}=\cclass{QCAM}$ (proved by Kobayashi \emph{et al.} in \cite{kobayashi}), but it is unlikely that $\cclass{QMAM}=\cclass{QAM}$, since $\cclass{QMAM}=\cclass{QIP}=\cclass{PSPACE}$ \cite{mw,qip=pspace}. 
 
 As we shall see in Section \ref{section:interactiveproofsforquantumstateisomorphism}, there is a barrier that prevents us from applying Theorem \ref{theorem:collapse} to \textsc{StateIsomorphism}: our quantum interactive proof systems for \textsc{StateNonIsomorphism} require quantum communication between verifier and prover. This prevents us from proving inclusion in QCAM. It is not clear that the problem admits such a proof system.
 
  However, if it is possible to produce an efficient classical description of the quantum states in the problem instance that is independent from how they are specified in the input, then it is possible to prove inclusion
  in QCAM. We show that this is the case for a restricted family of quantum states called \emph{stabilizer states}, a fact which allows us to prove the following. 
\begin{corollary}
	\label{theorem:productCollapse}
	If \textsc{StabilizerStateIsomorphism} is \emph{\cclass{QCMA}}-complete, then \emph{$\cclass{QPH}\subseteq\Sigma_2$}. 
\end{corollary}
Furthermore, the fact that stabilizer states can be described classically also implies the following.
\begin{theorem}
\textsc{StabilizerStateNonIsomorphism} is in $\text{\emph{QCSZK}}$.
\end{theorem}

Finally, we consider the state isomorphism problem for mixed quantum states. We show that this problem is QSZK-hard by reduction from the QSZK-complete problem of determining if a mixed state is product or separable. 
\begin{theorem}
\label{theorem:msiqszkhardNICE}
$(\epsilon,1-\epsilon)$-\textsc{MixedStateIsomorphism} is \cclass{QSZK}-hard.
\end{theorem}

While these state isomorphism problems all have classical certificates, we have been able to demonstrate that the complexity of each problem depends precisely on the inherent computational difficulty of working with the input states. Stabilizer states form one end of the spectrum: with a polynomial number of measurements a classical description can be produced. The other extreme is the mixed states, these are so computationally difficult to work with that it is not clear that \textsc{MixedStateIsomorphism} even belongs in \cclass{QMA}; even the problem of testing equivalence of two such states is \cclass{QSZK}-complete (see \cite{qszk}). Between these two extremes we have \textsc{StateIsomorphism}. While such states can be efficiently processed by a quantum circuit, and isomorphism can be certified classically, the analysis in Section \ref{section:interactiveproofsforquantumstateisomorphism} uncovers an interesting caveat. It seems that the ability to communicate quantum states is still required when we wish to check \emph{non}-isomorphism by interacting with a prover, or perhaps even to certify isomorphism with statistical zero knowledge. We thus draw attention to the following open question: can our protocols be modified to use exclusively classical communication?

The fact that an efficient quantum algorithm for \textsc{StateIsomorphism} would also yield one for \textsc{GraphIsomorphism}, combined with
Corollary \ref{corollary:qcma}, gives weight to the idea that this problem can be thought of as a candidate for a \cclass{QCMA}-intermediate problem. The fact that there are problems ``in between'' \cclass{BQP} and \cclass{QCMA}, and furthermore, that such problems are obtained by generalising \textsc{StringIsomorphism} suggests an interesting parallel between the classical and quantum classes.

In Section \ref{section:preliminariesanddefinitions} we give an overview of the tools and notation we will use for the rest of the paper. We also define the key problems and complexity classes we will be working with and prove some initial results that we build on later. In Section \ref{section:interactiveproofsforquantumstateisomorphism} we demonstrate quantum interactive proof systems for the \textsc{StateIsomorphism} problems. In Section \ref{section:aquantumpolynomialhierarchy} we define a notion of a quantum polynomial hierarchy, and prove the hierarchy collapse results.

\section{Preliminaries and definitions}
\label{section:preliminariesanddefinitions}

Recall that quantum states are represented by unit trace positive semi-definite operators $\rho$ on a Hilbert space $\mathcal{H}$ called the \emph{state space} of the system. A state is \emph{pure} if $\rho^2=\rho$. Otherwise, we say that the state is \emph{mixed}. 
By definition then, for any pure state $\rho$ on $\mathcal{H}$ we have that $\rho=|\psi\rangle\langle\psi|$ for some unit vector $|\psi\rangle\in\mathcal{H}$, and we refer to pure states by their corresponding \emph{state vector} $|\psi\rangle$ (which is unique up to multiplication by a phase).
Mixed states are convex combinations of the outer products of some set of state vectors
$
\rho=\sum_{i} p_i|\psi_i\rangle\langle\psi_i|.
$
In what follows we refer to the Hilbert space $\mathbb{C}^2$ by $\mathcal{H}_2$. Recall that an $n$-qubit pure state $|\psi\rangle\in\mathcal{H}_2^{\otimes n}$ is \emph{product} if
$|\psi\rangle=|\psi_1\rangle\otimes\cdots\otimes|\psi_n\rangle
$
where $\otimes$ denotes tensor product and for all $i$, $|\psi_i\rangle\in\mathcal{H}_2$. For any bitstring $x_1\dots x_n\in\{0,1\}^n$, we say that $|x\rangle=\otimes_{i=1}^n|x_i\rangle$ is a computational basis state.

A useful measure of the distinguishability of a pair of quantum states is the \emph{trace distance}. Let $\rho,\sigma$ be quantum states with the same state space. Their trace distance is the quantity
$D(\rho,\sigma)=\frac{1}{2}\lVert \rho-\sigma\rVert_1,
$ where $\lVert M \rVert_1=\text{tr}[|M|]$ is the trace norm.

We say that a quantum circuit $Q$ \emph{accepts} a state $|\psi\rangle$ if measuring the first qubit of the state $Q|\psi\rangle$ in the computational basis yields outcome $1$. We say that the circuit \emph{rejects} the state otherwise.
Let $X$ be an index set. We say that a uniform family of quantum circuits $\{Q_x~:~x\in X\}$ is \emph{polynomial-time generated} if there exists a polynomial-time Turing machine that takes as input $x\in X$ and halts with an efficient description of the circuit $Q_x$ on its tape. Such a definition neatly captures the notion of an efficient quantum computation \cite{watrous}. 

\begin{figure}[h!]
\centering
\includegraphics[scale=1.0]{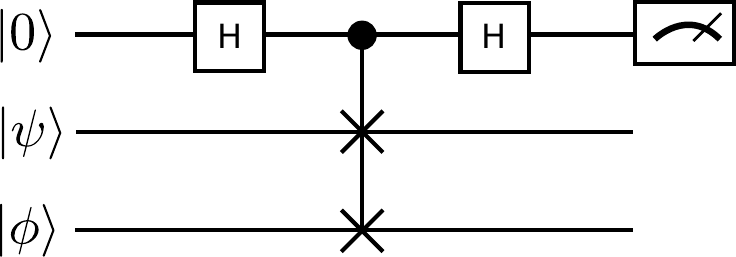}
\caption{The SWAP test circuit.}
\label{fig:swap}
\end{figure}

 We make use of a quantum circuit known as the \emph{SWAP test} \cite{swaptest}, illustrated in Figure \ref{fig:swap}. This circuit takes as input pure states $|\psi\rangle,|\phi\rangle$ and accepts (denoted $T(|\psi\rangle,|\phi\rangle)=1$) with probability $(1+|\langle\psi|\phi\rangle|^2)/2$. Note that $T(|\psi\rangle,|\phi\rangle)=1$ with probability $1$ if $|\psi\rangle=e^{i\tau}|\phi\rangle$ for some $\tau\in[-2\pi,2\pi]$, but is equal to $1$ with probability $1/2$ if they are orthogonal.
  The SWAP test can be therefore be used as an efficient quantum algorithm for testing if two quantum states are equivalent. In what follows we use some notation from complexity theory and formal language theory. In particular, if a problem $A$ is polynomial-time many-one reducible to a problem $B$ we denote this by $A \le_p B$. We denote by $\{0,1\}^n$ the set of bitstrings of length $n$, furthermore, $\{0,1\}^*$ denotes the set of all bitstrings. For a bitstring $x$, we denote by $|x|$ the length of the bitstring. We say that a function $f:\mathbb{N}\rightarrow [0,1]$ is \emph{negligible} if for every constant $c$ there exists $n_c$ such that for all $n\ge n_c$, $f(n)<1/n^c$. We use the shorthand $f(n)=\text{poly}(n)$ (\emph{resp.} $f(n)=\exp(n)$) to state that $f$ scales as a polynomially bounded (exponentially bounded) function in $n$.
  
  A decision problem is a set of bitstrings $A\subseteq\{0,1\}^*$. An algorithm is said to decide $A$ if for all $x\in\{0,1\}^*$ it outputs YES if $x\in A$ and NO otherwise. In quantum computational complexity it is useful to use the less well known notion of a \emph{promise problem} to allow for more control over problem instances. A promise problem is a pair of sets $(A_{\text{YES}},A_{\text{NO}})\subseteq\{0,1\}^*\times \{0,1\}^*$ such that $A_{\text{YES}}\cap A_{\text{NO}}=\emptyset$. An algorithm is said to decide $(A_{\text{YES}},A_{\text{NO}})$ if for all $x\in A_{\text{YES}}$ it outputs YES and for all $x\in A_{\text{NO}}$ it outputs NO. Note that the algorithm is not required to do anything in the case where an input $x$ does not belong to $A_{\text{YES}}$ or $A_\text{NO}$.

\subsection{Quantum Merlin-Arthur, Quantum Arthur-Merlin}
For convenience, we give a number of definitions related to quantum generalisations of public coin proof systems. In particular, we focus on Quantum Arthur-Merlin (\cclass{QAM}) and Quantum Merlin-Arthur, the quantum versions of AM and MA respectively. We use the definitions in \cite{watrous,mw} as our guide.

\begin{definition}[\cclass{QMA}]
	A promise problem $A=\promprob{A}$ is in $\cclass{QMA}(a,b)$ for functions $a,b:\mathbb{N}\rightarrow[0,1]$ if there exists a polynomial-time generated uniform family of quantum circuits $\{V_{x}~:~x\in\{0,1\}^*\}$ and polynomially bounded $p:\mathbb{N}\rightarrow\mathbb{N}$ such that
	\begin{itemize}
	\item for all $x\in\promprobY{A}$ there exists $|\psi\rangle\in\mathcal{H}_2^{\otimes p(|x|)}$ such that
	\begin{align*}
	\pr[V_x \text{ accepts } |\psi\rangle] \ge a(|x|);
	\end{align*}
	\item for all $x\in\promprobN{A}$ and for all $|\psi\rangle\in\mathcal{H}_2^{\otimes p(|x|)}$,
		\begin{align*}
		\pr[V_x \text{ accepts } |\psi\rangle] \le b(|x|).
		\end{align*}
	\end{itemize}
\end{definition}
The class \cclass{QCMA} is defined in the same way, but with the restriction that the certificate $|\psi\rangle$ must be a computational basis state $|x\rangle$.

A \emph{\cclass{QAM} verification procedure} is a tuple $(V,m,s)$ where \begin{align*}V=\{V_{x,y}~:~x\in\{0,1\}^*,y\in\{0,1\}^{s(|x|)}\}\end{align*} is a uniform family of polynomial time generated quantum circuits, and $m,s:\mathbb{N}\rightarrow\mathbb{N}$ are polynomially bounded functions. Each circuit acts on $m(|x|)$ qubits sent by Merlin, and $k(|x|)$ qubits which correspond to Arthur's workspace. For all $x,y$, we say that $V_{x,y}$ accepts (\emph{resp.} rejects) a state $|\psi\rangle\in\mathcal{H}_2^{\otimes m(|x|)}$ if, upon measuring the first qubit of the state \begin{align*}V_{x,y} |\psi\rangle|0\rangle^{\otimes k(|x|)}\end{align*} in the standard basis, the outcome is `$1$' (\emph{resp.} `$0$').
\begin{definition}[\cclass{QAM}]
	A promise problem $A=\promprob{A}$ is in $\cclass{QAM}(a,b)$ for functions $a,b:\mathbb{N}\rightarrow[0,1]$ if there exists a \cclass{QAM} verification procedure $(V,m,s)$ such that
	\begin{itemize}
		\item for all $x\in \promprobY{A}$, there exists a collection of $m(|x|)$-qubit quantum states $\{|\psi_{y}\rangle\}$ such that
		\begin{align*}
		\frac{1}{2^{s(|x|)}}\sum_{y\in\{0,1\}^{s(|x|)}}\pr[V_{x,y} \text{ accepts } |\psi_y\rangle]\ge a(|x|);
		\end{align*}
		\item for all $x\in \promprobN{A}$, and for all collections of $m(|x|)$-qubit quantum states $\{|\psi_{y}\rangle\}$, it holds that
		\begin{align*}
		\frac{1}{2^{s(|x|)}}\sum_{y\in\{0,1\}^{s(|x|)}}\pr[V_{x,y} \text{ accepts } |\psi_y\rangle]\le b(|x|).
		\end{align*}
	\end{itemize}
\end{definition}
The class \cclass{QCAM} is defined in the same way but with the states $\{|\psi_y\rangle\}$ restricted to computational basis states. The class \cclass{QCMAM} is similar, but has an extra round of interaction.

\begin{definition}[\cclass{QCMAM}]
	A promise problem $A=\promprob{A}$ is in $\cclass{QCMAM}(a,b)$ for functions $a,b:\mathbb{N}\rightarrow[0,1]$ if there exists a $\cclass{QAM}$ verification procedure $(V,m,s)$ and a polynomially bounded function $p:\mathbb{N}\rightarrow\mathbb{N}$ such that
	\begin{itemize}
		\item for all $x\in \promprobY{A}$, there is a certificate bitstring $c\in\{0,1\}^{p(|x|)}$ and a collection of length $m(|x|)$ bitstrings $\{z^c_{y}\}$ such that
		\begin{align*}
		\frac{1}{2^{s(|x|)}}\sum_{y\in\{0,1\}^{s(|x|)}}\pr[V_{x,y} \text{ accepts } |c\rangle\otimes|z^c_y\rangle]\ge a(|x|);
		\end{align*}
		\item for all $x\in\promprobN{A}$, all certificate bitstrings $c\in\{0,1\}^{p(|x|)}$ and all collections of length $m(|x|)$ bitstrings $\{z^c_{y}\}$, it holds that
		\begin{align*}
		\frac{1}{2^{s(|x|)}}\sum_{y\in\{0,1\}^{s(|x|)}}\pr[V_{x,y} \text{ accepts } |c\rangle\otimes|z^c_y\rangle]\le b(|x|).
		\end{align*}
	\end{itemize}
\end{definition}

\subsection{Quantum interactive proofs and zero knowledge}
An interactive proof system consists of a \emph{verifier} and a \emph{prover}. The computationally unbounded prover attempts to convince the computationally limited verifier that a particular statement is true. A quantum interactive proof system is where the verifier is equipped with a quantum computer, and quantum information can be transferred between verifier and prover. Our formal definitions will follow those of Watrous \cite{qszk,watrous}.

A \emph{quantum verifier} is a polynomial time computable function $V$, where for each $x\in\{0,1\}^*$, $V(x)$ is an efficient classical description of a sequence of quantum circuits $V(x)_1,\dots,V(x)_{k(|x|)}$. Each circuit in the sequence acts on $v(|x|)$ qubits that make up the verifier's private workspace, and a buffer of $c(|x|)$ communication qubits that both verifier and prover have read/write access to.

A \emph{quantum prover} is a function $P$ where for each $x\in\{0,1\}^*$, $P(x)$ is a sequence of quantum circuits $P(x)_1,\dots P(x)_{l(|x|)}$. Each circuit in the sequence acts on $p(|x|)$ qubits that make up the prover's private workspace, and the $c(|x|)$ communication qubits that are shared with each verifier circuit. Note that no restrictions are placed on the circuits $P(x)$, since we wish the prover to be computationally unbounded. We say that a verifier $V$ and a prover $P$ are \emph{compatible} if all their circuits act on the same number of communication qubits, and if for all $x\in\{0,1\}^*$, $k(|x|)=\lfloor m(|x|)/2+1\rfloor$ and $l(|x|)=\lfloor m(|x|)/2+1/2\rfloor$, for some $m(|x|)$ which is taken to be the number of messages exchanged between the prover and verifier. We say that $(P,V)$ are a compatible $m$-message prover-verifier pair.

Given some compatible $m$-message prover-verifier pair $(P,V)$, we define the quantum circuit
\begin{align*}
(P(x),V(x)):=\begin{cases}V(x)_1\cdot P(x)_1\dots P(x)_{m(|x|)/2}\cdot V(x)_{m(|x|)/2+1}&\text{ if } m(|x|) \text{ is even,}\\
P(x)_1\cdot V(x)_1\dots P(x)_{(m(|x|)+1)/2}\cdot V(x)_{(m(|x|)+1)/2}&\text{ if } m(|x|) \text{ is odd.}
\end{cases}
\end{align*}
Let $q(|x|)=p(|x|)+c(|x|)+v(|x|)$. We say that $(P,V)$ accepts an input $x\in\{0,1\}^*$ if the result of measuring the verifier's first workspace qubit of the state
\begin{align*}
(P(x),V(x))|0^{q(|x|)}\rangle
\end{align*}
in the computational basis is $1$, and that it rejects the input if the measurement result is $0$.
\begin{definition}[$\cclass{QIP}(k)$]
Let $M=\promprob{M}$ be a promise problem, let $a,b:\mathbb{N}\rightarrow[0,1]$be functions and $k\in\mathbb{N}$. Then $M\in \cclass{QIP}(k)(a,b)$ if and only if
there exists a $k$-message verifier $V$ such that
\begin{itemize}
\item if $x\in \promprobY{M}$ then 
\begin{align*}
\max_{P}\left(\pr[(P,V) \text{ accepts } x]\right) \ge a(|x|),
\end{align*}
\item if $x\in \promprobN{M}$ then 
\begin{align*}
\max_P \left(\pr[(P,V) \text{ accepts } x]\right) \le b(|x|),
\end{align*}
where the maximisation is performed over all compatible $k$-message provers. We say that the pair $(P,V)$ is an interactive proof system for $M$.
\end{itemize}
\end{definition}
Let us now define what it means for a quantum interactive proof system to be \emph{statistical zero-knowledge}. Define the function 
\begin{align*}
\text{view}_{V,P}(x,j)=\text{tr}_P[(P(x),V(x))_j|0^{q(|x|)}\rangle\langle 0^{q(|x|)}|(P(x),V(x))_j^\dagger],
\end{align*}
where $(P(x),V(x))_j$ is the circuit obtained from running $(P(x),V(x))$ up to the $j^{\text{th}}$ message.
For some index set $X$, we say that a set of density operators $\{\rho_x~:~x\in X\}$ is \emph{polynomial-time preparable} if there exists a polynomial-time uniformly generated family of quantum circuits $\{Q_x~:~x\in X\}$, each with a designated set of output qubits, such that for all $x\in X$, the state of the output qubits after running $Q_x$ on a canonical initial state $|0\rangle^{\otimes n}$ is equal to $\rho_x$.
\begin{definition}[Honest Verifier Quantum Statistical Zero-Knowledge (HVQSZK)]
Let $M=\promprob{M}$ be a promise problem, let $a,b:\mathbb{N}\rightarrow[0,1]$ and $k:\mathbb{N}\rightarrow\mathbb{N}$ be functions. Then $M\in \cclass{HVQSZK}(k)(a,b)$ if and only if $M\in \cclass{QIP}(k)(a,b)$ with quantum interactive proof system $(P,V)$ such that there exists a polynomial-time preparable set of density operators $\{\sigma_{x,i}\}$ such that for all $x\in\{0,1\}^*$, if $x\in \promprobY{M}$ then
\begin{align*}
D(\sigma_{x,i},\text{\emph{view}}_{P,V}(x,i))\le \delta(|x|)
\end{align*}
for some negligible function $\delta$.
\end{definition}

It is known that the class of problems that have quantum statistical zero knowledge proof systems (QSZK) is equivalent to the class of problems that have honest verifier quantum statistical zero knowledge proof systems (HVQSZK) \cite{qszk}. Therefore, we refer to HVQSZK as QSZK, and only consider honest verifiers.

In the next section we give a formal definition of \textsc{StringIsomorphism}.

\subsection{Permutations and \textsc{StringIsomorphism}}
\label{subsection:permutationsandstringisomorphism}
Let $\Omega$ be a finite set. A bijection $\sigma:\Omega\rightarrow\Omega$ is called a \emph{permutation} of the set $\Omega$. 
 The set of all permutations of a finite set $\Omega$ forms a group under composition. This group is called the \emph{symmetric group}, and we denote it by $\mathfrak{S}(\Omega)$. For $x\in\Omega$ and $\sigma\in \mathfrak{S}(\Omega),$ we denote the image of $x$ under $\sigma$ by $\sigma(x)$.

A \emph{string} $\mathfrak{s}:\Omega\rightarrow\Sigma$ is an assignment of \emph{letters} from a finite set $\Sigma$ called an \emph{alphabet} to the elements of a finite \emph{index set} $\Omega$. Let $\mathfrak{s}:\Omega\rightarrow\Sigma$ be a string. The letters of $\mathfrak{s}$ are indexed by elements of the index set $\Omega$. The letter corresponding to $i\in\Omega$ is thus denoted by $\mathfrak{s}_{i}$. Let $\sigma\in \mathfrak{S}(\Omega)$ be a permutation. Then the action of $\sigma$ on $\mathfrak{s}$ is denoted by $\sigma(\mathfrak{s})$, and is a string such that for all $i\in\Omega$,
$\sigma(\mathfrak{s})_{i}=\mathfrak{s}_{\sigma(i)}.
$
In this paper we often deal with permutations of strings indexed by natural numbers. Hence, we denote the symmetric group  $\mathfrak{S}([n])$ by $\mathfrak{S}_n$, where $[n]:=\{1,\dots,n\}$. In what follows we denote the fact that a group $G$ is a subgroup of a group $H$ by $G\le H$. The following decision problem is related to \textsc{GraphIsomorphism}\cite{luks, babai}, and forms the basis of our work.

\begin{problem}
	\textsc{StringIsomorphism}\\
	\textit{Input:} \emph{Finite sets $\Omega,\Sigma$, a permutation group $G\le \mathfrak{S}(\Omega)$ specified by a set of generators, and strings $\mathfrak{s},\mathfrak{t}:\Omega\rightarrow\Sigma$.}\\
	\textit{Output:} \textsc{Yes} \emph{if and only if there exists $\sigma\in G$ such that} 
	$
	\sigma(\mathfrak{s})=\mathfrak{t}.$
\end{problem}

It is clear that \textsc{StringIsomorphism} is at least as hard as \textsc{GraphIsomorphism}: a polynomial time many-one reduction can be obtained from \textsc{GraphIsomorphism} by ``flattening'' the adjacency matrices of the graphs in question into bitstrings. The set of string permutations that correspond to graph isomorphisms form a proper subgroup of the full symmetric group. Indeed the  algorithm in \cite{babai} is actually an algorithm for \textsc{StringIsomorphism}, which solves \textsc{GraphIsomorphism} as a special case.
\subsection{Stabilizer states}
The Gottesmann-Knill theorem \cite{gknill} states that any quantum circuit made up of CNOT, Hadamard and phase gates along with single qubit measurements can be simulated in polynomial time by a classical algorithm. Such circuits are called stabilizer circuits, and any $n$-qubit quantum state $|\psi\rangle$ such that $|\psi\rangle=Q|0\rangle^{\otimes n}$ for a stabilizer circuit $Q$ is referred to as a \emph{stabilizer state}. 

Let $|\psi\rangle$ be an $n$-qubit state. A unitary $U$ is said to be a \emph{stabilizer} of $|\psi\rangle$ if $U|\psi\rangle=\pm 1|\psi\rangle$. The set of stabilizers of a state $|\psi\rangle$ forms a finite group under composition called the \emph{stabilizer group} of $|\psi\rangle$, denoted $\text{Stab}(|\psi\rangle)$.

The Pauli matrices are the unitaries
\begin{align*}
\sigma_{00}:=\begin{pmatrix}
1&0\\0&1\end{pmatrix},~\sigma_{01}:=\begin{pmatrix}0&1\\1&0
\end{pmatrix},~\sigma_{10}:=\begin{pmatrix}
1&0\\0&-1
\end{pmatrix},~\sigma_{11}:=\begin{pmatrix}
0&-i\\i&0
\end{pmatrix},
\end{align*}
which form a finite group $\mathcal{P}$ under composition called the \emph{single qubit Pauli group}. The $n$-qubit Pauli group $\mathcal{P}_n$ is the group with elements $\{(\pm 1)U_1\otimes \dots\otimes (\pm 1 )U_n~:~U_j\in \mathcal{P}\}\cup \{(\pm i)U_1\otimes \dots\otimes (\pm i)U_n~:~U_j\in \mathcal{P}\}$.

It is well known (\emph{c.f.} \cite{ag} Theorem 1) that an $n$-qubit stabilizer state $|\psi\rangle$ is uniquely determined by the finite group $S(|\psi\rangle):=\text{Stab}(|\psi\rangle)\cap \mathcal{P}_n$, of size $2^n$. Hence, $|\psi\rangle$ is determined by the $n=\log(2^n)$ elements of $\mathcal{P}_n$ that generate $S(|\psi\rangle)$. These elements each take $2n$ bits to specify the Pauli matrices in the tensor product, and an extra bit to specify the overall $\pm 1$ phase. This fact, along with the following theorem, means that given a polynomial number of copies of a stabilizer state $|\psi\rangle$, we can produce an efficient classical description of that state by means of the generators of $S(|\psi\rangle)$.
\begin{theorem}[Montanaro \cite{mont}, corollary of Theorem $1$]
\label{theorem:classicaldescription}
There exists a quantum algorithm with the following properties:
\begin{itemize}
\item Given access to $O(n)$ copies of an $n$-qubit stabilizer state $|\psi\rangle$, the algorithm outputs a bitstring describing a set of $n$-qubit Pauli operators $s_1,\dots, s_n\in \mathcal{P}_n$ such that $\langle s_1,\dots s_n\rangle = S(|\psi\rangle)$;
\item the algorithm halts after $O(n^3)$ classical time steps;
\item all collective measurements are performed over at most two copies of the state $|\psi\rangle$;
\item the algorithm succeeds with probability $1-1/\exp(n)$.
\end{itemize}
\end{theorem}

\subsection{Permutations of quantum states and isomorphism}
\label{subsection:quantumstateisomorphism}
Let $\sigma\in\mathfrak{S}_n$ be a permutation. Then the following is a unitary map acting on $n$-partite states that implements it as a permutation of the subsystems (see \emph{e.g.} \cite{aram})
\begin{align}
\label{eq:harrowop}
P_\sigma:=\sum_{i_1,\dots,i_n\in[d]}|i_{\sigma(1)}\dots i_{\sigma(n)}\rangle\langle i_1,\dots i_n|.
\end{align}
Note that $P_\sigma$ depends on the dimensions of the subsystems of the $n$-partite states on which it acts. Nevertheless, here we will only consider quantum states where each subsystem is a qubit.

The focus of this work is on a number of variations on the following promise problem, \textsc{StateIsomorphism}. 

In what follows, let $\mathcal{Q}_{m,n}$ for $m\ge n$ denote the set of all quantum circuits with $m$ input qubits and $n$ output qubits. In particular, $Q_{n,n}$ is the set of all pure state quantum circuits on $n$ qubits. Then, for $m>n$, $Q_{m,n}$ is the set of all mixed state circuits that can be obtained by discarding the last $m-n$ output qubits of the circuits in $Q_{m,m}$.

When we specify a circuit with a subscript label, such as  $Q_\psi\in\mathcal{Q}_{m,n}$, we do so to easily refer to the state of the output qubits when the circuit is applied to the state $|0\rangle^m$. In particular, when $m=n$ this is the pure state $|\psi\rangle\in\mathbb{C}^{2^n}$, and the mixed state $\psi$ acting on $\mathbb{C}^{2^n}$ otherwise.

\begin{problem}
	\emph{
\textsc{StateIsomorphism} (SI)\\
\textit{Input}: Efficient descriptions of quantum circuits $Q_{\psi_0}$ and $Q_{\psi_1}$ in $\mathcal{Q}_{n,n}$, a set of permutations $\{\tau_1,\dots \tau_k\}$ generating some permutation group $\langle\tau_1\dots \tau_k\rangle=:G\le \mathfrak{S}_n$, and a function $\epsilon:\mathbb{N}\rightarrow[0,1]$ such that $\epsilon(n)\ge 1/\text{poly}(n)$ for all $n$.
\begin{itemize}
\item[] \textsc{YES}: There exists a permutation $\sigma \in G$ such that
$
|\langle \psi_1|P_\sigma |\psi_0\rangle| = 1.
$
\item[] \textsc{NO}: For all permutations $\sigma \in G$, 
$
|\langle \psi_1|P_\sigma |\psi_0\rangle| \le \epsilon(n).
$\end{itemize}
}
\end{problem}
The next problem is a special case of the above, defined in terms of stabilizer states.
\begin{problem}
	\emph{
\textsc{StabilizerStateIsomorphism} (SSI)\\
\textit{Input}: Efficient descriptions of quantum circuits $Q_{\psi_0}$ and $Q_{\psi_1}$ in $\mathcal{Q}_{n,n}$ such that $|\psi_0\rangle$ and $|\psi_1\rangle$ are stabilizer states, a set of permutations $\{\tau_1,\dots \tau_k\}$ generating some permutation group $\langle\tau_1\dots \tau_k\rangle=:G\le \mathfrak{S}_n$, and $\epsilon:\mathbb{N}\rightarrow[0,1]$ such that $\epsilon(n)\ge 1/\text{poly}(n)$ for all $n$.
\begin{itemize}
\item[] \textsc{YES}: There exists a permutation $\sigma \in G$ such that
$
|\langle \psi_1|P_\sigma |\psi_0\rangle| = 1.
$
\item[] \textsc{NO}: For all permutations $\sigma \in G$, 
$
|\langle \psi_1|P_\sigma |\psi_0\rangle| \le \epsilon(n).
$\end{itemize}
}
\end{problem}
Finally, we consider the state isomorphism problem for mixed states.
\begin{problem}
	\emph{
$(\epsilon,1-\epsilon)$-\textsc{MixedStateIsomorphism} (MSI)\\
\textit{Input}: Efficient descriptions of quantum circuits $Q_{\rho_0}$ and $Q_{\rho_1}$ in $\mathcal{Q}_{2n,n}$, a set of permutations $\{\tau_1,\dots \tau_k\}$ generating some permutation group $\langle\tau_1\dots \tau_k\rangle=:G\le \mathfrak{S}_n$, and $\epsilon:\mathbb{N}\rightarrow[0,1]$. 
\begin{itemize}
\item[] \textsc{YES}: There exists a permutation $\sigma \in G$ such that
$
D(P_\sigma\rho_0 P_\sigma^\dagger,\rho_1) \le \epsilon(n).
$
\item[] \textsc{NO}: For all permutations $\sigma \in G$, 
$
D(P_\sigma\rho_0 P_\sigma^\dagger,\rho_1) \ge 1-\epsilon(n).
$
\end{itemize}
}
\end{problem}
We also consider the above problems where the permutation group specified is equal to the symmetric group $G=\mathfrak{S}_n$. We denote these problems with the prefix $\mathfrak{S}_n$, for example, $\mathfrak{S}_n\text{-SI}$. It is clear that $\textsc{SSI}\le_p \textsc{SI}\le_p\textsc{MSI}$. We now show that SI is in \cclass{QCMA}.
\begin{proposition}
	$\textsc{StateIsomorphism}\in \cclass{QCMA}$.
\end{proposition}
\begin{proof}
In the case of a YES instance, there exists $\sigma\in G$ such that $|\langle \psi_1|P_\sigma|\psi_0\rangle|=1$. The latter equality can be verified by means of a SWAP-test on the states $P_\sigma|\psi_0\rangle$ and $|\psi_1\rangle$, which by definition will accept with probability equal to $1$. Since the states $|\psi_0\rangle$ and $|\psi_1\rangle$ are given as an efficient classical descriptions of quantum circuits that will prepare them, this verification can be performed in quantum polynomial time. Furthermore, there exists an efficient classical description of the permutation $\sigma$ in terms of the generators of the group specified in the input, each of which can be described via their permutation matrices. The unitary $P_\sigma$ can be implemented efficiently by Arthur given the description of $\sigma$. 

Determining membership/non-membership of some permutation $\sigma\in\mathfrak{S}_n$ in the permutation group $G\le \mathfrak{S}_n$ specified by the set of generators $\{\tau_1,\dots \tau_k\}$ can be verified in classical polynomial time by utilizing standard techniques from computational group theory. In particular, since we are considering permutation groups we can use the Schreier-Sims algorithm to obtain a base and a strong generating set for $G$ in polynomial time from $\{\tau_1,\dots,\tau_k\}$. These new objects can then be used to efficiently verify membership in $G$ \cite{sims, FHL, luks2}.

In the case that the states are not isomorphic, we have by definition that for all permutations $\sigma\in G$, $|\langle \psi_1|P_\sigma|\psi_0\rangle|\le \epsilon(n)$, which can again be verified by using the SWAP-test, which will accept the states with probability at most $1/2+\epsilon(n)$. 
\end{proof}

It is not clear if MSI is in \cclass{QCMA}, or even in \cclass{QMA}. While the isomorphism $\sigma$ can still be specified efficiently classically, it is not known if there exists an efficient quantum circuit for testing if two mixed states are close in trace distance. In fact, this problem is known as the \textsc{StateDistinguishability} problem, and is QSZK-complete \cite{qszk}. 

There exists a polynomial-time many-one reduction from \textsc{GraphIsomorphism} to \textsc{SSI}, indeed it is identical to the reduction from \textsc{GraphIsomorphism} to \textsc{StringIsomorphism}. \textsc{SSI} is in turn trivially reducible to the isomorphism problems for pure and mixed states respectively. These problems are therefore at least as hard as \textsc{GraphIsomorphism}. Interestingly however, there also exists a reduction from \textsc{GraphIsomorphism} to a restricted form of \textsc{SI} where the permutation group $G$ is equal to the full symmetric group $\mathfrak{S}_n$ (as stated earlier, we refer to this problem as $\mathfrak{S}_n$-$\textsc{StateIsomorphism}$). In order to demonstrate this, we require a family of quantum states referred to as \emph{graph states} \cite{graphstates}. Let $G=(V,E)$ be an $n$-vertex graph. For each vertex $v\in V$, define the observable
$
K^{(v)}:=\sigma_x^{(v)}\prod_{w\in N(v)}\sigma_z^{(w)}
$
where $N(v)$ is the neighborhood of $v$, and $\sigma_i^{(j)}$ denotes the $n$-qubit operator consisting of Pauli $\sigma_i$ applied to the $j^{\text{th}}$ qubit and identity on the rest. The graph state $|G\rangle$ is defined to be the state stabilized by the set $S_G:=\{K^{(v)}~:~v\in V\}$, that is,
$
K^{(v)}|G\rangle=|G\rangle
$
for all $v\in V$. Since the stabilizers of a graph state $|G\rangle$ are all elements of the $|V|$ qubit Pauli group, graph states are stabilizer states, and the following theorem provides an upper bound on the overlap of non-equal graph states.

\begin{theorem}[Aaronson-Gottesmann  \cite{ag}. See also \cite{gmc}, Theorem 8]
	\label{theorem:agottstab}
	Let $|\psi\rangle,|\phi\rangle$ be non-orthogonal stabiliser states, and let $s$ be the minimum, taken over all sets of generators $\{P_1,\dots P_n\}$ for $S(|\psi\rangle)$ and $\{Q_1,\dots Q_n\}$ for $S(|\phi\rangle)$, of the number of $i$ values such that $P_i\neq Q_i$. Then $|\langle\psi|\phi\rangle|=2^{-s/2}$.
\end{theorem}

We can now describe the reduction.
\begin{proposition}
$\textsc{GraphIsomorphism}\le_p \mathfrak{S}_n\text{-}\textsc{StateIsomorphism}$.
\end{proposition}
\begin{proof}
Consider two $n$-vertex graphs $G$ and $H$. If $G=H$ then clearly $|\langle G|H\rangle|^2=1$ since $|G\rangle$ and $|H\rangle$ are the same state up to a global phase. Suppose $G\neq H$. Then necessarily $s>0$, so by Theorem \ref{theorem:agottstab} we have that $|\langle G|H\rangle|^2\le \frac{1}{2}$.
Consider a permutation $\sigma\in\mathfrak{S}_n$. Then for each $v\in V$, $K^{(\sigma(v))}=P_\sigma K^{(v)} P_\sigma^T$, so $|\langle\sigma(G)|P_\sigma|G\rangle|^2=1$. Explicitly, if $G\cong H$ then there exists a permutation of the vertices $\sigma$ such that $\sigma(G)=H$ and so $|\langle \sigma(G)|H\rangle|^2=|\langle G|P_\sigma^T |H\rangle|^2=1$. If $G\not\cong H$ then for all $\sigma$, $\langle G|P_\sigma^T |H\rangle|^2\le \frac{1}{2}$.

To complete the reduction we must show that for any graph $G=(V,E)$, a description of a quantum circuit that prepares $|G\rangle$ can be produced efficiently classically. This is trivial, an alternate definition of graph states \cite{graphstates} gives us that $|G\rangle=\Pi_{\{i,j\}\in E}CZ_{ij}|+\rangle^{\otimes |V|}$, where $CZ_{ij}$ is the controlled-$\sigma_z$ operator with qubit $i$ as control and $j$ as output.
\end{proof}

Therefore, the \textsc{StateIsomorphism} problem where no restriction is placed on the permutations is at least as hard as \textsc{GraphIsomorphism}. This is in stark contrast to the complexity of the corresponding classical problem, which is trivially in \cclass{P}: two bitstrings are isomorphic under $\mathfrak{S}_n$ if and only if they have the same Hamming weight, which is easily determined.

\section{Interactive proof systems}\label{section:interactiveproofsforquantumstateisomorphism}
In this section we will prove Theorem \ref{theorem:SIQSZK}. To do so, we will first demonstrate a quantum interactive proof system for \textsc{StateNonIsomorphism} (SNI) with two messages. We then show that this quantum interactive proof system can be made statistical zero knowledge. In order to prove the former, we will require the following lemma.

\begin{lemma}[Harrow-Lin-Montanaro \cite{hlm}, Lemma 12]
	\label{theorem:harrowetal}
	Given access to a sequence of unitaries $U_1,\dots, U_n$, along with their inverses $U_1^\dagger,\dots, U_n^\dagger$ and controlled implementations c-$U_1$,\dots ,c-$U_n$, as well as the ability to produce copies of a state $|\psi\rangle$ promised that one of the following cases holds:
	\begin{enumerate}
	\item For some $i$, $U_i|\psi\rangle=|\psi\rangle$;
		\item For all $i$, $|\langle\psi|U_i|\psi\rangle|\le 1-\delta$.
	\end{enumerate}
	Then there exists a quantum algorithm which distinguishes between these cases using $O(\log n /\delta)$ copies of $|\psi\rangle$, succeeding with probability at least $2/3$.
\end{lemma}
We can now prove the following.
\begin{theorem}
\textsc{StateNonIsomorphism} is in \cclass{QIP}$(2)$.
\end{theorem}
\begin{proof}
We will prove that the following constitutes a two message quantum interactive proof system for SNI.
\begin{mdframed}
\begin{enumerate}
\item (V) Uniformly at random, select $\sigma\in G$ and $j\in\{0,1\}$. Send the state $|\Psi\rangle^{\otimes k}$ to the prover, where $k=O(\log(|G|)/(1-\epsilon(n)))$ and $|\Psi\rangle=P_\sigma|\psi_j\rangle$.
\item (P) Send $j'\in\{0,1\}$ to the verifier.
\item (V) Accept if and only if $j'=j$.
\end{enumerate}
\end{mdframed}
Obtaining a uniformly random element from $G$ as in step $1$ can be achieved efficiently if the verifier is in possession of a base and a strong generating set for $G$. These can be obtained in polynomial time from any generating set of $G$ by using Schreier-Sims algorithm \cite{sims, FHL, luks2}.
 For a permutation $\pi\in G$, we define the $2n$ qubit circuit $U^{(j)}_\pi=\text{SWAP}\cdot (P_{\pi^{-1}}\otimes P_\pi)$, where the SWAP acts so as to swap the two $n$ qubit states, that is, $\text{SWAP}|\psi_0\rangle|\psi_1\rangle=|\psi_1\rangle|\psi_0\rangle$. Now consider the sets of quantum circuits $C^{(j)}_G=\{U^{(j)}_{\pi}~:~\pi\in G\}$ for $j\in\{0,1\}$, each of cardinality $|G|$. Since each circuit in $C^{(0)}_G\cup C^{(1)}_G$ is made up two permutations and a SWAP gate, each of their inverses can easily be obtained. Additionally, the controlled versions of these gates can be implemented via standard techniques.

Consider first the YES case. The $k=O(\log(|G|)/(1-\epsilon(n)))$ copies of $|\Psi\rangle$ enables the prover to determine $j$ with success probability at least $2/3$ in the following manner.
\begin{mdframed}
\begin{enumerate}
\item Uniformly at random, select $j'\in\{0,1\}$.
\item Prepare $k$ copies of the state $|\Psi\rangle|\psi_{j'}\rangle$
\item Use the HLM algorithm with the state $|\Psi\rangle|\psi_{j'}\rangle$ and the set of circuits $C^{(j')}_G$ as input. If the algorithm reports case $1$ then output $j'$, otherwise output $j'\oplus 1$.
\end{enumerate}
\end{mdframed}
Let us check that the HLM algorithm will work for our purposes. In the case that the prover's guess is correct and $j'=j$, we have that $|\Psi\rangle|\psi_{j'}\rangle = (P_\sigma\otimes I)|\psi_j\rangle|\psi_j\rangle$, and so  \begin{align*}U_{\sigma}(P_\sigma\otimes I)|\psi_j\rangle|\psi_j\rangle&=
\text{SWAP}\cdot (P_{\sigma^{-1}}\otimes P_{\sigma})\cdot(P_\sigma\otimes I)|\psi_j\rangle|\psi_j\rangle\\
&=\text{SWAP}\cdot(I\otimes P_\sigma)|\psi_j\rangle|\psi_j\rangle\\
&=|\Psi\rangle|\psi_j\rangle.
\end{align*}
This corresponds to case $1$ of Lemma \ref{theorem:harrowetal}.
If the prover's guess is incorrect $j'\neq j$ then for all $\pi\in G$
\begin{align*}
|\langle \Psi |\langle \psi_{j'} |U_\pi|\Psi\rangle|\psi_{j'}\rangle|&=|\langle \Psi |\langle \psi_{j'} |\text{SWAP}\cdot(P_{\pi^{-1}}\otimes P_{\pi})(P_\sigma\otimes I)|\psi_j\rangle|\psi_{j'}\rangle|\\
&=|\langle\Psi|\langle\psi_{j'}|(P_{\pi}\otimes P_{\pi^{-1}\cdot\sigma})|\psi_{j'}\rangle|\psi_{j}\rangle|\\
&\le |\langle \psi_{j}|P_{\sigma}^\dagger P_{\pi}|\psi_{j'}\rangle|
\cdot
|\langle \psi_{j'}| P_{\pi^{-1}\cdot \sigma}|\psi_{j}\rangle|\\
&\le \epsilon(n)^2,
\end{align*}
with the last inequality following from the fact that we are in the YES case: for all $\sigma\in G$, we have that $|\langle\psi_2|P_\sigma|\psi_1\rangle|\le a(n)$. This corresponds to case $2$ of Lemma \ref{theorem:harrowetal}. Therefore, the HLM algorithm allows the prover to determine if her guess was correct or not, with success probability at least $2/3$.

Consider now the NO case, where we have that for some $\sigma\in G$, $|\langle\psi_1|P_\sigma|\psi_2\rangle|=1$. To determine $j$ correctly, a cheating prover must be able to distinguish the mixed states $\rho_j=\frac{1}{|G|}\sum_{\pi\in G}\left(P_\pi|\psi_j\rangle\langle\psi_j|P_\pi^\dagger\right)^{\otimes k}$ correctly for $j\in\{1,2\}$, when given $k$ copies. However,
\begin{align*}
\lVert \rho_1-\rho_2\rVert_1 &= \frac{1}{|G|}\left\lVert \sum_{\pi\in G}P_\pi^{\otimes k}\left(|\psi_1\rangle\langle\psi_1|\right)^{\otimes k}P_\pi^{\dagger\otimes k}-\sum_{\pi\in G}P_\pi^{\otimes k}\left(|\psi_2\rangle\langle\psi_2|\right)^{\otimes k}P_\pi^{\dagger\otimes k}\right\rVert_1\\
&= \frac{1}{|G|}\left\lVert \sum_{\pi\in G}P_\pi^{\otimes k}P_\sigma^{\otimes k}\left(|\psi_1\rangle\langle\psi_1|\right)^{\otimes k}P_\sigma^{\dagger\otimes k}P_\pi^{\dagger\otimes k}-\sum_{\pi\in G}P_\pi^{\otimes k}\left(|\psi_2\rangle\langle\psi_2|\right)^{\otimes k}P_\pi^{\dagger\otimes k}\right\rVert_1\\
&= \frac{1}{|G|}\left\lVert \sum_{\pi\in G}P_\pi^{\otimes k}\left(|\psi_2\rangle\langle\psi_2|\right)^{\otimes k}P_\pi^{\dagger\otimes k}-\sum_{\pi\in G}P_\pi^{\otimes k}\left(|\psi_2\rangle\langle\psi_2|\right)^{\otimes k}P_\pi^{\dagger\otimes k}\right\rVert_1\\
&=0,
\end{align*}
so they are indistinguishable. Note that the fact that the prover has been given $k$ copies does not help, as the overlap is $0$. In this case, the probability that the prover can guess $j$ correctly is therefore equal to $1/2$.
\end{proof}

We can use a standard amplification argument to modify the above protocol so that it has negligible completeness error, which means that it can be made statistical zero knowledge. We prove this now.

\begin{theorem}
\textsc{StateNonIsomorphism} is in \emph{QSZK}.
\end{theorem}
\begin{proof}
We first show that the protocol above can be modified to have exponentially small completeness error. This allows us to show that the protocol is quantum statistical zero knowledge. 

First, the verifier sends the prover $k'=O(n\log(|G|)/(1-\epsilon(n)))$ copies of the state $|\Psi\rangle$. The prover can then use HLM $n$ times to guess $j$, responding with the value of $j$ that appears in $n/2$ or more of the trials. Let $X_1\dots X_n\in\{\text{`T'},\text{`F'}\}$ be the set of independent random variables corresponding to whether or not the prover guessed correctly on the $i^\text{th}$ repetition. By Lemma \ref{theorem:harrowetal}, we have that $\pr[X_i=\text{`T'}]\ge2/3$ and so 
\begin{align*}
\pr\left[\text{Prover guesses correctly}\right]&=1-\pr\left[\frac{1}{n}\sum_{i=1}^n X_i < 1/2\right]\\
&=1-\pr\left[\frac{1}{n}\sum_{i=1}^n X_i-2/3<-1/6\right]\\
&\ge 1-2^{-\Omega(n)}
\end{align*}
via the Chernoff bound (explicitly, for $p,q\in[0,1]$, we have that $\pr\left[\sum_{i=1}^n (X_i-p)/n<-q \right]<e^{-q^2n/2p(1-p)}$). Clearly, sending $k'$ copies of $|\Psi\rangle$ rather than $k$ gives no advantage to the prover, the trace distance between the mixed states $\rho_0$ and $\rho_1$ is still $0$ in the NO case.

What remains is to show that the protocol is statistical zero knowledge. This is easily obtained, and follows by similar reasoning to the protocol in \cite{qszk}: the view of the verifier after the first step can be obtained by the simulator by selecting $\sigma$ and $j$ then preparing $k'$ copies of the state $|\Psi\rangle$. The view of the verifier after the prover's response can be obtained by tracing out the message qubits and supplying the verifier with the value $j$. Since completeness error is exponentially small, the trace distance between the simulated view and the actual view is a negligible function.
\end{proof}

If we change (relax) the condition for the two states to be non isomorphic (NO instance) to: `There exists $\sigma \in G$ such that $|\langle\psi_2|P_\sigma|\psi_1\rangle|\geq b(n)$' then the distance between the two states $\rho_j=\frac{1}{|G|}\sum_{\pi\in G}\left(P_\pi|\psi_j\rangle\langle\psi_j|P_\pi^\dagger\right)^{\otimes k}$ for $j\in\{1,2\}$ is upper bounded by \begin{align*}
\lVert \rho_1-\rho_2\rVert_1
&=\frac{1}{|G|}\left\lVert \sum_{\pi\in G}\left(P_{\pi}\right)^{\otimes k}\left(P_\sigma|\psi_1\rangle\langle\psi_1|P_{\sigma}^\dagger - |\psi_2\rangle\langle\psi_2|\right)^{\otimes k} \left(P_\pi^\dagger\right)^{\otimes k}\right\rVert_1\\
&\leq \frac{1}{|G|}\sum_{\pi\in G} \left\lVert \left(P_{\pi}\right)^{\otimes k}\left(P_\sigma|\psi_1\rangle\langle\psi_1|P_{\sigma}^\dagger - |\psi_2\rangle\langle\psi_2|\right)^{\otimes k} \left(P_\pi^\dagger\right)^{\otimes k}\right\rVert_1\\
&= \left\lVert \left(P_\sigma|\psi_1\rangle\langle\psi_1|P_{\sigma}^\dagger - |\psi_2\rangle\langle\psi_2|\right)^{\otimes k}\right\rVert_1\\
&=2 \sqrt {1- \left|\langle\psi_1|P_{\sigma}^\dagger  |\psi_2\rangle\right|^{2 k}}\leq 2 \sqrt {1- \epsilon(n)^{2 k}}, \end{align*}
where first inequality is just triangular inequality, last inequality follows from the promise and last equality is just rewriting the trace distance for pure states in terms of their scalar product. Now, putting the value of $k=\frac {\log n}{1-a(n)}$, algebraic manipulations and using the fact that $\log(1-x)>-2x$ for all $x\in (0,1/2)$, we get, for any $b(n)\in (1/2,1)$,
\begin{align*}
\lVert \rho_1-\rho_2\rVert_1 
&= 2 \sqrt {1- b(n)^{\frac {2 \log n}{1-a(n)}}}= 2 \sqrt {1- n^{\frac {2\log {b(n)}}{1-a(n)} }}\\
&\leq  2 \sqrt {1- n^{\frac {-4(1-b(n))}{1-a(n)} }}.\end{align*}
\\
Then the maximal probability of distinguishing between these two states is upper bounded by
\begin{align*}
p\leq 1/2+\sqrt{1- n^{\frac {-4(1-b(n))}{1-a(n)} }}.
\end{align*}
We have thus proved Theorem \ref{theorem:SIQSZK}. Corollary \ref{corollary:qcma} follows easily: if SI was QCMA-complete then all QCMA problems would be reducible to it, and would belong in QSZK. 
While SI belongs in QCMA, the above protocol requires quantum communication. It is not clear if a similar protocol exists that uses classical communication only. In the next theorem we show that such a protocol exists for \textsc{StabilizerStateNonIsomorphism}, since stabilizer states can be described efficiently classically.
\begin{theorem}
\textsc{StabilizerStateNonIsomorphism} is in $\emph{QCSZK}$.
\end{theorem}
\begin{proof}
It suffices to show that the state $|\Psi\rangle$ in the protocol above can be communicated to the prover using classical communication only. We know from Theorem \ref{theorem:classicaldescription} that a classical description can be obtained efficiently from $O(n)$ copies of $|\Psi\rangle$. These copies can be prepared efficiently, since they are specified in the problem instance by quantum circuits that prepare them.
\end{proof}
We now prove that \textsc{MixedStateIsomorphism} is QSZK-hard (Theorem \ref{theorem:msiqszkhardNICE}). We actually prove the following stronger result.
\begin{theorem}
	\label{theorem:msiqszk}
	$(\epsilon,1-\epsilon)$-$\mathfrak{S}_n$-\textsc{MixedStateIsomorphism} is \emph{\cclass{QSZK}}-hard for all $\epsilon(n)=1/\exp(n)$.
\end{theorem}
We prove this by reduction from the following problem $(\alpha,\beta)$-\textsc{ProductState}, which as shown in \cite{septesting} is \cclass{QSZK}-hard even when $\alpha=\epsilon,\beta=1-\epsilon$ and $\epsilon$ goes exponentially small in $n$.

\begin{problem}
	$(\alpha,\beta)\text{-}\textsc{ProductState}$\\
	\emph{
		\textit{Input}: Efficient description of a quantum circuit $Q_{\rho}$ in $\mathcal{Q}_{0,n}$.\\
		\textit{YES:} There exists an $n$-partite product state $\sigma_1\otimes \cdots\otimes\sigma_n$ such that
		$D(\rho,\sigma_1\otimes\cdots\otimes\sigma_n)\le\alpha\\$
		\textit{NO:} For all $n$-partite product states $\sigma_1\otimes \cdots\otimes\sigma_n$,
$D(\rho,\sigma_1\otimes\cdots\otimes\sigma_n)\ge\beta.
		$}
\end{problem}
We make use of the following lemma. 
For an $n$-partite mixed state $\rho$, let $\rho_i$ denote the state of the $i^{\text{th}}$ subsystem, obtained by tracing out the other subsystems.
\begin{lemma}[Gutoski et al. \cite{septesting}, Lemma 7.4]
	Let $\rho$ be an $n$ qubit state. If there exists a product state $\sigma_1\otimes\cdots\otimes\sigma_n$ such that
	$\lVert\rho-\sigma_1\otimes\cdots\otimes \sigma_n\rVert_1\le \alpha,
	$ then
	$\lVert\rho-\rho_1\otimes\cdots\otimes \rho_n\rVert_1\le (n+1)\alpha
	$
\end{lemma}
\begin{figure}
\centering
\includegraphics{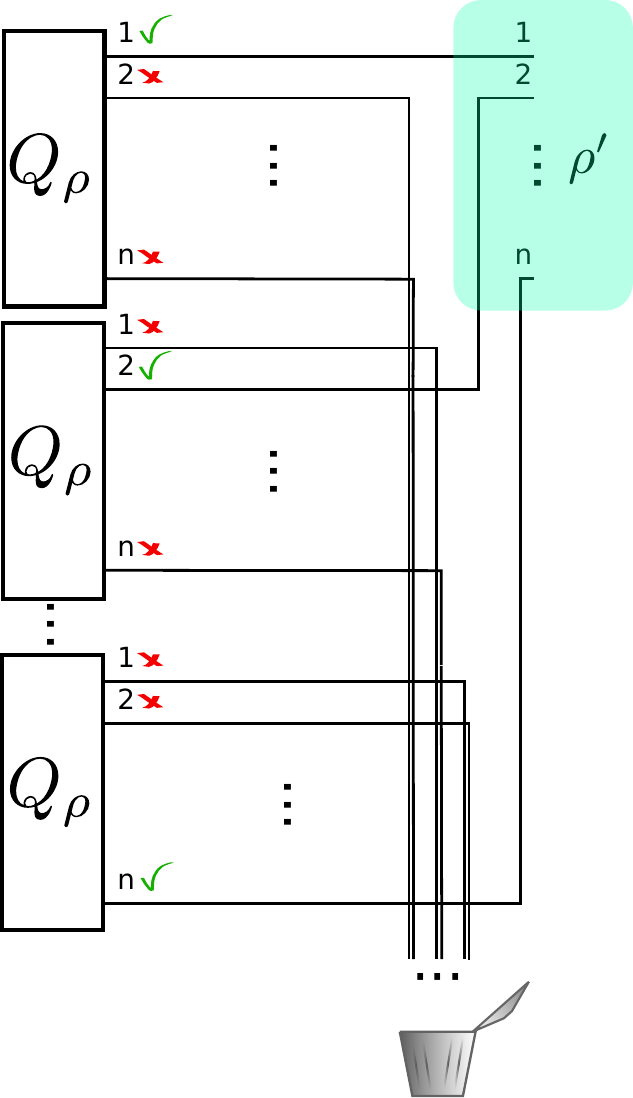}

\caption{Constructing the state $\rho'=\rho_1\otimes\dots\otimes\rho_n$ from $n$ copies of the input circuit $Q_\rho$.}
\label{fig:qcirc}
\end{figure}
\begin{proof}[Proof of Theorem \ref{theorem:msiqszk}.]
	We now must show that every instance of $(\alpha,\beta)$-\textsc{ProductState} can be converted to an instance of $(\alpha',\beta')$-$\mathfrak{S}_n$-\textsc{MixedStateIsomorphism}. In particular, consider an instance $\rho$ of $(\alpha,\beta)$-\textsc{ProductState}. Our reduction takes this to an instance $(\rho,\rho')$ of $((n+1)\alpha,\beta)$-$\mathfrak{S}_n$-\textsc{MixedStateIsomorphism}, where $\rho'=\rho_1\otimes\dots\otimes \rho_n$ can be prepared in the following way from $n$ copies of the state $\rho$. Denote these $n$ copies as $\rho^{(1)},\dots,\rho^{(n)}$. The $i^{\text{th}}$ qubit line of $\rho'$ is the $i^{\text{th}}$ qubit line of $\rho^{(i)}$, all unused qubit lines are discarded (illustrated in Figure \ref{fig:qcirc}).
	
	Let $\rho$ be an $n$-partite state. If $\rho$ is product then 
$
	D(\rho,\rho_1\otimes \cdots \otimes \rho_n) \le (n+1)\alpha/2
$
	and so $(\rho,\rho')$ correspond to a YES instance of $((n+1)\alpha,\beta)$-$\mathfrak{S}_n$-\textsc{MixedStateIsomorphism}. If $\rho$ is a NO instance of $(\alpha,\beta)$-\textsc{ProductState} then $D(\rho,\theta)\ge \beta$ for all product states $\theta$. This means that $D(\rho,P_\sigma\rho_1\otimes\cdots\otimes\rho_n P_\sigma)\ge \beta$ for all $\sigma\in\mathfrak{S}_n$ since all such states are product.
\end{proof}

In this section we have shown that \textsc{StateIsomorphism} is in QSZK, and so is unlikely to be QCMA-complete unless all problems in QCMA have quantum statistical zero knowledge proof systems. We have also shown that \textsc{StabilizerStateIsomorphism} has a quantum statistical zero knowledge proof system that uses classical communication only, and that \textsc{MixedStateIsomorphism} is QSZK-hard.

In the next section, we show that the quantum polynomial hierarchy collapses if \textsc{StabilizerStateIsomorphism} is QCMA-complete.
\section{A quantum polynomial hierarchy}
\label{section:aquantumpolynomialhierarchy}
Yamakami \cite{yamakami} considers a more general framework of quantum complexity theory, where computational problems are specified with quantum states as inputs, rather than just classical bitstrings. We find that using this more general view of computational problems makes it easier to define a very general quantum polynomial-time hierarchy, which can then be ``pulled back'' to a hierarchy that has more conventional complexity classes (\emph{e.g.} \cclass{BQP}, \cclass{QMA}) as its lowest levels.

Following \cite{yamakami} we consider classes of \emph{quantum promise problems}, where the YES and NO sets are made up of quantum states. We use the work's notion of quantum $\exists$ and $\forall$ complexity class operators in our definitions. These yield classes that are more general than we need, so we use restricted versions where all instances are computational basis states. 

Let $|\psi\rangle\in\mathcal{H}_2^{\otimes n}$ be an $n$-qubit state. Then in analogy to the length of a classical bitstring $|x_1\dots x_n|=n$, we define the length of the state $|\psi\rangle$ as $\big| |\psi\rangle\big|=n$. The set $\{0,1\}^*:=\cup_{i=1}^{\infty}\{0,1\}^i$ is the set of all bitstrings. Analogously, the set
$\mathcal{H}_2^*:= \bigcup_{i=1}^{\infty}\mathcal{H}_2^{\otimes i}
$
is the set of all qubit states. A \emph{quantum promise problem} is therefore a pair of sets $\mathcal{A}_{\text{YES}},\mathcal{A}_{\text{NO}}\subseteq \mathcal{H}_2^*$ with $\mathcal{A}_{\text{YES}}\cap\mathcal{A}_{\text{NO}}=\emptyset$. Note that to differentiate quantum promise problems from the traditional definition with bitstrings, we use the calligraphic font. We make use of the following complexity class, made up of quantum promise problems.
\begin{definition}[$\cclass{BQP}^q$]
	A quantum promise problem $\qpromprob{A}$ is in the class $\cclass{BQP}^q(a,b)$, for functions $a,b:\mathbb{N}\rightarrow [0,1]$ if there exists a polynomial-time generated uniform family of quantum circuits $\{Q_{n}~:~n\in\mathbb{N}\}$ such that for all $|\psi\rangle\in\mathcal{H}_2^*$
	\begin{itemize}
		\item if $|\psi\rangle\in \qpromprobY{A}$ then 
		$\pr[Q_{l} \text{ accepts } |\psi\rangle]\ge a(l);
		$\item if $|\psi\rangle\in \qpromprobN{A}$ then 
		$\pr[Q_{l} \text{ accepts } |\psi\rangle]\le b(l),
		$
	\end{itemize}
	where $l=\big||\psi\rangle\big|$.
\end{definition}
Classes made up of quantum promise problems will always be denoted with the `q' superscript. It is clear that $\cclass{BQP}\subseteq\cclass{BQP}^q$, because any classical promise problem can be converted to a quantum promise problem by considering bitstrings as computational basis states. There is nothing to be gained computationally by imposing that inputs are expressed as computational basis states rather than bitstrings, so we make no distinction between the ``bitstring promise problems'' and the ``computational basis state'' promise problems. Indeed let $\cclass{C}^q$ be a quantum promise problem class. Then we define \begin{align*}\cclass{C}:=\{\mathcal{A}\in\cclass{C}^q~:~\text{all states in } \qpromprobY{A} \text{ and } \qpromprobN{A} \text{ are computational basis states.}\}\end{align*}
The classes $\cclass{BQP}^q$ and $\cclass{BQP}$ are related in this way.
 For the remainder of this work we will assume that all complexity classes are made up of quantum promise problems. It will be convenient for us to consider even conventional complexity classes such as \cclass{QMA} and \cclass{QCMA} to be defined with problem instances specified as computational basis states, rather than as bitstrings. Defining them in this way does not affect the classes in any meaningful way, but it is useful for our purposes. In particular, instead of referring to instances of a promise problem $x\in A_{\text{YES}}\cup A_{\text{NO}}$, we will refer to computational basis states in a quantum promise problem $|x\rangle\in\mathcal{A}_{\text{YES}}\cup \mathcal{A}_{\text{NO}}$.
 
 The following operators are well known from classical complexity theory, and are adapted here for quantum promise problem classes.

\begin{definition}[$\exists_s$/$\forall_s$ operator]
	Let $\cclass{C}$ be a complexity class. A promise problem $\qpromprob{A}$ is in  $\exists_s\cclass{C}$ for $s\in\{q,c\}$ if there exists a promise problem $\qpromprob{B}\in\cclass{C}$ and a polynomially bounded function $p:\mathbb{N}\rightarrow\mathbb{N}$ such that 
	\begin{align*}
	\qpromprobY{A}=\{|\psi\rangle\in\mathcal{H}_2^*~:~\exists |y\rangle\in S~|\psi\rangle\otimes |y\rangle\in \qpromprobY{B}\},
	\end{align*}
	and 
	\begin{align*}
	\qpromprobN{A}=\{|\psi\rangle\in\mathcal{H}_2^*~:~\forall |y\rangle\in S~|\psi\rangle\otimes |y\rangle\in \qpromprobN{B}\},
	\end{align*}
	where the set $S$ is equal to $\{|x\rangle~:~x\in \{0,1\}^{p(||\psi\rangle|)}\}$ if $s=c$, and $\mathcal{H}_2^{\otimes p(||\psi\rangle|)}$ if $s=q$. 
	The class $\forall_s \cclass{C}$ is defined analogously, but with the quantifiers swapped.
\end{definition}

We can now define the quantum polynomial hierarchy.

\begin{definition}[Quantum polynomial time hierarchy]
	Let $\Sigma^q_0=\Pi^q_0=\cclass{BQP}^q$. For $k\ge 1$, let $s_1\dots s_k\in\{c,q\}^k$. Then
	\begin{align*}
	s_1\dots s_k\text{-}\Sigma_k^q=\exists_{s_1}s_2\cdots s_k\text{-}\Pi_{k-1}^q
	\end{align*}
	and
	\begin{align*}
	s_1\dots s_k\text{-}\Pi_k^q=
\forall_{s_1}s_2\cdots s_k\text{-}\Sigma_{k-1}^q
	\end{align*}
\end{definition}

This definition leads to complexity classes that include promise problems with quantum inputs. Such classes are not well understood, so we do not use this hierarchy in its full generality. Instead we take each level $\Sigma_i^q$ or $\Pi_i^q$, and strip out all problems except those defined in terms of computational basis states by using $\Sigma_i$ or $\Pi_i$. Doing so makes familiar classes emerge, indeed it is clear that $\Sigma_0=\Pi_0=\cclass{BQP}$, $\text{c-}\Sigma_1=\cclass{QCMA}$ and $\text{q-}\Sigma_1=\cclass{QMA}$. This provides a generalisation of the ideas of Gharibian and Kempe \cite{gk} into a full hierarchy: our definition of the class $\text{cq-}\Sigma_2$ corresponds directly to theirs. For our purposes we require the following technical lemma.
\begin{lemma}
	\label{lemma:eahousekeeping}
For all $k$, let $C_k=s\text{-}\Sigma_k^q$ or $C_k=s\text{-}\Pi_k^q$ for any $s\in\{q,c\}^k$. Then
\begin{enumerate}
	\item \label{item:ececeqec} $\exists_c\exists_c C_k=\exists_c C_k$
	\item \label{item:acaceqac}
	$\forall_c\forall_c C_k=\forall_c C_k$
	\item \label{item:ecsubeq} $\exists_c C_k\subseteq\exists_q C_k$
	\item \label{item:acsubaq}
	$\forall_c C_k\subseteq\forall_q C_k$
	\item \label{item:eceqequeq}
	$\exists_c\exists_q C_k=\exists_q\exists_c C_k=\exists_q C_k$
	\item \label{item:acaqequaq}
	$\forall_c\forall_q C_k=\forall_q\forall_c C_k=\forall_q C_k$
\end{enumerate}
\end{lemma}
\begin{proof}
	(\ref{item:ececeqec}) and (\ref{item:acaceqac}) are trivial.
	(\ref{item:ecsubeq}) follows because a \cclass{BQP} verifier circuit can force all certificates to be classical by measuring each qubit in the standard basis before processing. (\ref{item:acsubaq}) follows because this class is complementary.
	(\ref{item:eceqequeq}) follows by a similar argument: take $\qpromprob{A}\in\exists_c\exists_q C_k$, where the classical certificate is of length $p_1(|x|)$, and the quantum certificate is of length $p_2(|x|)$. Clearly $\qpromprob{A}$ is in $\exists_q C_k$ with certificate length $p_1(|x|)+p_2(|x|)$, since the first $p_1(|x|)$ qubits can be measured before processing, so that they are forced to be computational basis states. The other direction, $\exists_q C_k\subseteq \exists_c\exists_q C_k$, follows trivially by setting the classical certificate length to $0$. Then (\ref{item:acaqequaq}) follows from (\ref{item:eceqequeq}) because the classes are complementary.   
\end{proof}

\subsection{Quantum hierarchy collapse}
\label{section:quantumhierarchycollapse}
Our main focus in this paper is on problems in \cclass{QCMA}. Therefore, it is sufficient to adopt the definition of the hierarchy with all certificates classical. Let
$
\cclass{QPH}^q:=\bigcup_{i=0}^{\infty}\text{cc}\cdots \text{c-}\Sigma_{i}^q.
$
We consider the restricted hierarchy $\cclass{QPH}$, (\emph{N.B.}, without the `q' superscript). Since each certificate is classical, when we refer to classes at each level we omit the certificate specification, referring to each level as simply $\Sigma_i$ or $\Pi_i$. Also, note that we are considering the computational basis state restriction of each level of the hierarchy so we omit the `q' superscript.
We make use of the following lemmas.
\begin{lemma}
For all $i\ge 1$, $\exists_c \Sigma_i=\Sigma_i$ and 
$\forall_c\Pi_i=\Pi_i$.
\end{lemma}
\begin{proof}
Both follow as corollaries of Lemma \ref{lemma:eahousekeeping}, parts (\ref{item:ececeqec}) and (\ref{item:acaceqac}). 
\end{proof}
\begin{lemma}
\label{lemma:subsetcollapse}
	For all $i\ge 1$, if $\Sigma_{i}\subseteq\Pi_{i}$ or $\Pi_{i}\subseteq\Sigma_{i}$ then \emph{$\cclass{QPH}\subseteq \Sigma_i$}.
\end{lemma}
\begin{proof}
We prove first that if the equality $\Sigma_i=\Pi_i$ held for some $i\ge 1$ then for all $j>i$, $\Sigma_j\subseteq \Sigma_i$. We prove this by induction on $j$. Consider the base case $j=i+1$. By definition, if $\mathcal{A}\in\Sigma_{i+1}$ then $\mathcal{A}\in \exists_c\Pi_i=\exists_c\Sigma_i=\Sigma_i$. Assume for the induction hypothesis that if $\Sigma_i=\Pi_i$ then $\Sigma_{j}\subseteq \Sigma_i$. Let $k=j-i+1$. For $k$ odd and $\mathcal{A}\in\Sigma_{j+1}$ we have that $\mathcal{A}\in\underbrace{\exists_c \forall_c\cdots\exists_c}_{k}\Pi_i=\underbrace{\exists_c \forall_c\cdots\exists_c}_{k}\Sigma_i=\underbrace{\exists_c\forall_c\cdots\forall_c}_{k-1}\Sigma_i=\Sigma_j$. By the induction hypothesis this is a subclass of $\Sigma_i$. The case for even $k$ follows in the 
same way.
Since for all $i\ge 0$, $\Sigma_{i}=\text{co-}\Pi_{i}$, we have that if $\Sigma_{i}\subseteq\Pi_{i}$ or $\Pi_{i}\subseteq\Sigma_{i}$ then $\Sigma_{i}=\Pi_{i}$, and so the hierarchy collapses.
\end{proof}

The following two propositions are important for our purposes, and can be proved using similar techniques to those used in the proofs of $\cclass{AM}=\cclass{BP}\cdot \cclass{NP}$ and $\cclass{AM}\subseteq \Pi_2^P$.
We emphasise that the latter is
in terms of the \emph{quantum} polynomial hierarchy, indeed it would be remarkable if a similar result held for in terms of the classical hierarchy. The proofs follow in Sections \ref{subsection:qam=bpqma} and \ref{subsection:bqpmaupperbound}.

\begin{proposition}
\label{prop:bpqam}
\emph{$\cclass{QCAM}\subseteq\cclass{BP}\cdot \cclass{QCMA}$}, and \emph{$\cclass{QAM}\subseteq\cclass{BP}\cdot \cclass{QMA}$}.
\end{proposition}
A corollary of this is the following.
\begin{proposition}
\label{prop:qampi2}
	\emph{$\cclass{QCAM}\subseteq \text{cc-}\Pi_2$}, and \emph{$\cclass{QAM}\subseteq \text{cq-}\Pi_2$}.
\end{proposition}

In what follows, we will refer to the class \cclass{QCMAM}: a generalisation of $\cclass{QCAM}$ which has an extra round of interaction between Arthur and Merlin.
Kobayashi \emph{et al.} \cite{kobayashi} show that this class is equal to \cclass{QCAM}.
\begin{theorem}[Kobayashi-Le Gall-Nishimura \cite{kobayashi}, Theorem 7 (iv)]
	\label{theorem:kobayashi}
 \emph{$\cclass{QCMAM}=\cclass{QCAM}$}.
\end{theorem}
The next proposition uses this fact, and allows us to complete the proof of Theorem \ref{theorem:collapse}.
\begin{proposition}
\label{prop:coqcmaqcam}
If \emph{$\cclass{co}$-$\cclass{QCMA}\subseteq \cclass{QCAM}$} then  \emph{$\cclass{QPH}\subseteq\cclass{QCAM}\subseteq\Pi_2$}.
\end{proposition}
\begin{proof}
Let $\mathcal{A}=\qpromprob{A}\in\Sigma_{2}$. Then by definition there exists a promise problem $\mathcal{B}=\qpromprob{B}\in\Pi_{1}=\cclass{co}$-$\cclass{QCMA}$ and a polynomially bounded function $p$ such that for all $|x\rangle\in\qpromprobY{A}$,
\begin{align}
\label{eq:A1}
\exists y\in\{0,1\}^{p(|x|)}|x\rangle\otimes |y\rangle\in\qpromprobY{B},
\end{align}
and for all $|x\rangle\in\qpromprobN{A}$,
\begin{align}
\label{eq:A2}
\forall y\in\{0,1\}^{p(|x|)}|x\rangle\otimes |y\rangle\in\qpromprobN{B}.
\end{align}

If co-$\cclass{QCMA}\subseteq \cclass{QCAM}$ then $\mathcal{B}\in \cclass{QCAM}$. The existentially (Eq. \ref{eq:A1}) and universally (Eq. \ref{eq:A2}) quantified $y$'s can be thought of as certificate strings, and so $\mathcal{A}\in \cclass{QCMAM}$. By Theorem \ref{theorem:kobayashi}, $\cclass{QCMAM}=\cclass{QCAM}$, and so $\mathcal{A}\in\Pi_2$. Hence, $\Sigma_2\subseteq\Pi_2$, and the hierarchy collapses to the second level by Lemma \ref{lemma:subsetcollapse}.
\end{proof}
We now have the tools we need to prove \ref{theorem:collapse}.
\begin{proof}[Proof of Theorem \ref{theorem:collapse}]
Suppose $\mathcal{A}\in \cclass{QCMA}\cap\cclass{co}$-$\cclass{QCAM}$. If $\mathcal{A}$ is \cclass{QCMA}-complete then this implies that $\cclass{QCMA}\subseteq\cclass{co}$-$\cclass{QCAM}$, equivalently $\cclass{co}$-$\cclass{QCMA}\subseteq\cclass{QCAM}$. The hierarchy then collapses to the second level via Proposition \ref{prop:coqcmaqcam}.
\end{proof}

We may now finish this section by providing evidence that \textsc{StabilizerStateIsomorphism} is not \cclass{QCMA}-complete, encapsulated in Corollary \ref{theorem:productCollapse}. We do this by proving the following.
\begin{proposition}
\textsc{StabilizerStateNonIsomorphism} is in \emph{\cclass{QCAM}}.
\end{proposition}
\begin{proof}
For a stabilizer state $|\psi\rangle$, denote by $s_{\psi}^{(1)},\dots, s_{\psi}^{(n)}\in \{\pm I,\pm X,\pm Y,\pm Z\}^n$ the classical strings that describe the stabilizer generators of $|\psi\rangle$ that we can obtain efficiently using the algorithm of Theorem \ref{theorem:classicaldescription}. We denote by $s_\psi$ the length $2n$ string that is obtained by concatenating these stabilizer strings, that is $s_{\psi} = s_{\psi}^{(1)}\dots s_{\psi}^{(n)}$. Then for any permutation $\sigma\in\mathfrak{S}_n$, we take $\sigma(s_{\psi}) = s_{\psi}^{(\sigma(1))},\dots, s_{\psi}^{(\sigma(n))}$. For a permutation group $G\le \mathfrak{S}_n$, consider the set
\begin{align*}
S_G := \bigcup_{j\in\{0,1\},\sigma\in G}\left\{\left(\sigma\left(s_{\psi_{j}}\right),\pi\right)~:~\pi\in G\land \sigma\left(s_{\psi_{j}}\right) = \sigma\left(s_{\psi_{j}}\right)\right\}.
\end{align*}
If there exists $\sigma$ such that $|\langle \psi_1|P_\sigma|\psi_0\rangle| = 1$ then $\sigma(s_{\psi_0}) = s_{\psi_1}$, and so in this case $|S_G| = |G|$. If for all $\sigma\in G$ we have that $|\langle \psi_1|P_\sigma|\psi_0\rangle|\le 1-\epsilon(n)$ then likewise for all $\sigma\in G$, $\sigma(s_{\psi_0}) \neq s_{\psi_1}$ and therefore $|S_G| = 2|G|$. If we can show that membership in $S_G$ can be efficiently verified by Arthur then we can apply the Goldwasser-Sipser set lower bound protocol \cite{gs} to determine isomorphism of the states. To convince Arthur with high probability that $(\sigma(s_{\psi_j}),\pi)\in S_G$, Merlin sends the permutation $\sigma$ and the index $j\in\{0,1\}$. Arthur can then obtain the string $s_{\psi_{j}}$ with probability greater than $1-1/\text{exp}(n)$ using Montanaro's algorithm of Theorem \ref{theorem:classicaldescription} applied to $U_{\psi_j}|0\rangle$. He can then verify in polynomial time that the string he received is equal to $\sigma(s_{\psi_{j}})$, that $\pi$ is an automorphism of $\sigma(s_{\psi_{j}})$, and that the permutation $\sigma$ is in the group $G$.
\end{proof}

We have provided evidence that SSI can be thought of as an intermediate problem for \cclass{QCMA}. In particular, we have shown that if it were in BQP, then \textsc{GraphIsomorphism} would also be in BQP, and furthermore, that its \cclass{QCMA}-completeness would collapse the quantum polynomial hierarchy. Such evidence is unfortunately currently out of reach for \textsc{StateIsomorphism}, because we have been unable to show that \textsc{StateNonIsomorphism} is in \cclass{QCAM}. Perhaps Arthur and Merlin must always use quantum communication if Arthur is to be convinced that two states are NOT isomorphic. This would be interesting, because he can be convinced that they \emph{are} isomorphic using classical communication only ($\textsc{StateIsomorphism}\in \cclass{QCMA}$).

\subsection{Proof of Proposition \ref{prop:bpqam}}
\label{subsection:qam=bpqma}
We begin by giving a definition of the BP complexity class operator. Note that we are still working in terms of the quantum promise problems defined earlier, which is clear from the use of the calligraphic font $\mathcal{A}$. In the following we take $x\sim X$ to mean that $x$ is an element drawn uniformly at random from a finite set $X$.
\begin{definition}[BP operator]
	Let $\cclass{C}$ be a complexity class. A promise problem $\qpromprob{A}$ is in $\cclass{BP}(a,b)\cdot\cclass{C}$ for functions $a,b:\mathbb{N}\rightarrow[0,1]$ if there exists $\qpromprob{B}\in \cclass{C}$ and a polynomially bounded function $p:\mathbb{N}\rightarrow\mathbb{N}$ such that
	\begin{itemize}
		\item For all $|\psi\rangle\in\qpromprobY{A}$,
		\begin{align*}
		\mathop{\pr}_{y\sim \{0,1\}^{p(|x|)}}[|\psi\rangle\otimes |y\rangle\in \qpromprobY{B}]\ge a(||\psi\rangle|);
		\end{align*}
		\item For all $|\psi\rangle\in\qpromprobN{A}$,
		\begin{align*}
		\mathop{\pr}_{y\sim \{0,1\}^{p(|x|)}}[|\psi\rangle\otimes |y\rangle\notin \qpromprobN{B}]\le b(||\psi\rangle|).
		\end{align*}
	\end{itemize}
\end{definition}
It is clear that the probabilities $a,b$ can be amplified in the usual way by repeating the protocol a sufficient number of times and taking a majority vote.
Let $(\{V_{x,y}\},m,s)$ be a \cclass{QAM} verification procedure. In what follows we make use of the functions \begin{align*}\mu(m,V_{x,y}):=\max_{|\psi\rangle\in\mathcal{H}_2^{\otimes \text{m}(|x|)}}\left(\pr[V_{x,y} \text{ accepts } |\psi\rangle]\right)
\end{align*}
and \begin{align*}\nu(m,V_{x,y}):=\min_{|\psi\rangle\in\mathcal{H}_2^{\otimes \text{m}(|x|)}}\left(\pr[V_{x,y} \text{ rejects } |\psi\rangle]\right).
\end{align*}
The following results of Marriott and Watrous \cite{mw} are useful for our purposes.
\begin{theorem}[Marriott-Watrous \cite{mw}, Theorem 4.2]
	\label{theorem:errorReduction}
	Let $a,b:\mathbb{N}\rightarrow[0,1]$ and polynomially bounded $q:\mathbb{N}\rightarrow[0,1]$ satisfy
	\begin{align*}
	a(n)-b(n)\ge \frac{1}{q(n)}
	\end{align*}
	for all $n\in \mathbb{N}$. Then \emph{$\cclass{QAM}(a,b)$ $\subseteq\cclass{QAM}(1-2^{-r},2^{-r})$}, for all polynomially bounded $r:\mathbb{N}\rightarrow[0,1]$.
\end{theorem}
\begin{proposition}[Marriott-Watrous \cite{mw}, Proposition 4.3]
	\label{prop:altFormulation}
	Let \begin{align*}
	\left(\left\{V_{x,y}~:~x\in\{0,1\}^*,y\in\{0,1\}^{s(|x|)}\right\},m:\mathbb{N}\rightarrow\mathbb{N},s:\mathbb{N}\rightarrow\mathbb{N}\right)
	\end{align*}
	be a \emph{\cclass{QAM}} verification procedure for a promise problem $A$ with completeness and soundness errors bounded by $1/9$. Then for any $x\in\{0,1\}^*$ and for $y\in\{0,1\}^{s(|x|)}$ chosen uniformly at random,
	\begin{itemize}
		\item if $|x\rangle\in \qpromprobY{A}$ then $\pr[\mu(m,V_{x,y})\ge 2/3]\ge 2/3$;
		\item if $|x\rangle\in \qpromprobN{A}$ then $\pr[\mu(m,V_{x,y})\le 1/3]\ge 2/3$.
	\end{itemize}
\end{proposition}

We can use these tools to prove Proposition \ref{prop:bpqam}. We prove it for \cclass{QAM}, the result follows for QCAM by similar reasoning.
\begin{proof}[Proof of Proposition \ref{prop:bpqam}]
	Suppose $\mathcal{A}=\qpromprob{A}\in\cclass{QAM}(a,b)$. By Theorem \ref{theorem:errorReduction}, there exists a $\cclass{QAM}$ verification procedure ($\{V_{x,y}\},m,s$) with completeness and soundness errors bounded by $1/9$. Thus by Proposition \ref{prop:altFormulation} we know that for all $x\in\{0,1\}^{*}$,
	if $|x\rangle\in \qpromprobY{A}$ then
	\begin{align*}
	\mathop{\pr}_{y\sim\{0,1\}^{s(|x|)}}[\mu(m,V_{x,y})\ge 2/3]\ge 2/3,
	\end{align*}
	which means that
	\begin{align*}
	\frac{1}{2^{s(|x|)}}\left|\left\{y\in\{0,1\}^{s(|x|)}~:~\exists |z\rangle\in\mathcal{H}_2^{\otimes m(|x|)}\pr[V_{x,y} \text{ accepts }|z\rangle]\ge 2/3\right\}\right|\ge 2/3.
	\end{align*}
	By similar reasoning, if $|x\rangle\in \qpromprobN{A}$ then
	\begin{align*}
	\frac{1}{2^{s(|x|)}}\left|\left\{y\in\{0,1\}^{s(|x|)}~:~\forall |z\rangle\in\mathcal{H}_2^{\otimes m(|x|)}\pr[V_{x,y} \text{ accepts }|z\rangle]\le 1/3\right\}\right|\ge 2/3.
	\end{align*}
	These conditions are precisely the conditions for a promise problem to belong in \cclass{QMA}. This means we can fix some promise problem $\qpromprob{B}\in\cclass{QMA}(2/3,1/3)$ and re-express these statements in the following form:
	\begin{itemize}
		\item if $|x\rangle\in\qpromprobY{A}$ then
		\begin{align*}
		\mathop{\pr}_{y\sim\{0,1\}^{s(|x|)}}[|x\rangle\otimes |y\rangle\in \qpromprobY{B}]\ge 2/3
		\end{align*}
		\item if $|x\rangle\in \qpromprobN{A}$ then
		\begin{align*}
		\mathop{\pr}_{y\sim\{0,1\}^{s(|x|)}}[|x\rangle\otimes |y\rangle\in \qpromprobN{B}]\ge 2/3,
		\end{align*}
	\end{itemize}
	and so $\mathcal{A}\in \cclass{BP}(2/3,1/3)\cdot\cclass{QMA}(2/3,1/3)$.
	
\end{proof}

\subsection{Proof of Proposition \ref{prop:qampi2}}
\label{subsection:bqpmaupperbound}
The following well known lemmas allow us to put $\cclass{BP}\cdot \cclass{QMA}$ (\emph{resp.} $\cclass{BP}\cdot\cclass{QCMA}$), and thus $\cclass{QAM}$ (\emph{resp.} \cclass{QCAM}), in the second level of the quantum polynomial-time hierarchy. We follow \cite{ab} but recast them in a more helpful form for our purposes. For a set of bitstrings $S\subseteq\{0,1\}^m$ and $x\in\{0,1\}^m$, we take $S\oplus x=\{s\oplus x~:~s\in S\}$.
\begin{lemma}
	Let $S\subseteq\{0,1\}^m$ for $m\ge 1$ such that
	\begin{align*}
	|S|\ge (1-2^{-k})\cdot 2^m,
	\end{align*}
	for $2^k\ge m$. Then there exists $t_1,\dots, t_m\in\{0,1\}^m$ such that
	\begin{align*}
	\bigcup_{i=1}^m S\oplus t_i=\{0,1\}^m.
	\end{align*}
\end{lemma}
\begin{proof}
	We prove this via the probabilistic method. Consider uniformly random $t_1,\dots, t_m\in\{0,1\}^m$. Then 
	\begin{align*}
	\mathop{\pr}_{r\sim\{0,1\}^m}\left[r\notin \bigcup_{i=1}^m S\oplus t_i\right]&=\prod_{i=1}^m\mathop{\pr}_{r\sim\{0,1\}^m}\left[r\notin S\oplus t_i\right]\le 2 ^{-km}.
	\end{align*}
	Consider the probability that there exists some $v\in\{0,1\}^m$ such that $v\notin \bigcup_{i=1}^m S\oplus t_i$,
	\begin{align*}
	\mathop{\pr}[\exists v\in\{0,1\}^m.v\notin \bigcup_{i=1}^m S\oplus t_i]&\le \sum_{i=1}^{2^m} 2^{-km}\\
	&=\frac{2^m}{2^{km}}\\
	&<1.
	\end{align*}
	Hence,
	\begin{align*}
	\pr\left[\bigcup_{i=1}^m S\oplus t_i=\{0,1\}^m\right]>0,
	\end{align*}
	and so there must exist $t_1,\dots, t_m$ as required.
\end{proof}
This yields the following corollary.
\begin{corollary}
	\label{lemma:quantSimEA}
	Let $S\subseteq\{0,1\}^m$ for $m\ge 1$ such that
	\begin{align*}
	|S|\ge (1-2^{-k})\cdot 2^m,
	\end{align*}
	for $2^k\ge m$. Then there exists $t_1,\dots, t_m$ such that for all $v\in\{0,1\}^m$, there exists $i\in[m]$ such that $t_i\oplus v\in S$.
\end{corollary}
We also require the following lemma, which comes from the opposite direction.
\begin{lemma}
	\label{lemma:quantSimAE}
	Let $S\subseteq\{0,1\}^m$ for $m\ge 1$ such that
	\begin{align*}
	|S|\ge (1-2^{-k})\cdot 2^m,
	\end{align*}
	for $2^k\ge m$. Then for all $t_1,\dots, t_m\in\{0,1\}^m$, there exists $v\in \{0,1\}^m$ such that $\bigwedge_{i\in[m]} \left(u_i\oplus v\in S\right)$.
\end{lemma}
\begin{proof}
	Assume that there exists $t_1\dots t_m$ such that for all $v\in\{0,1\}^m$ there exists $i\in[m]$ with $t_i\oplus v\notin S$. This implies that
	there exists $i\in\{1,\dots,m\}$ such that, for at least $2^m/m$ elements $v\in\{0,1\}^m$, we have that $t_i\oplus v\notin S$. Then
	\begin{align*}
	|S|<2^m-2^m/m=2^m(1-1/m)\le (1-2^{-k})\cdot 2^m,
	\end{align*}
	contradicting our assumption about the cardinality of $S$.
\end{proof}
We can now prove the Proposition \ref{prop:bpqam}. We prove it for $\cclass{BP}\cdot\cclass{QMA}$; the result for $\cclass{BP}\cdot\cclass{QCMA}$ follows in the same way.
\begin{proof}[Proof of Proposition \ref{prop:bpqam}]
	Let $\qpromprob{A}\in \cclass{BP}\cdot\cclass{QMA}$. Then by definition there exists $\qpromprob{B}\in \cclass{QMA}$ and polynomially bounded $p,r:\mathbb{N}\rightarrow\mathbb{N}$ such that if $|x\rangle\in\qpromprobY{A}$, \begin{align*}
	\mathop{\pr}_{y\sim \{0,1\}^{p(|x|)}}[|x\rangle\otimes |y\rangle\in \promprobY{B}]\ge 1-2^{-r(|x|)}.
	\end{align*}
	Set $S_x=\{y\in\{0,1\}^{p(|x|)}~:~|x\rangle\circ |y\rangle\in \qpromprobY{B}\}$. Then $|x\rangle\in \qpromprobY{B}$ implies that $|S_x|\ge (1-2^{r(|x|)})\cdot 2^{p(|x|)}$. By amplification of BP, we can choose $r$ to be whatever we want, so we choose it such that $2^{r(|x|)}\ge p(|x|)$. Then by Lemma \ref{lemma:quantSimAE}, 
	\begin{align}
	\label{eq:longAE}
	x\in \promprobY{A}\implies \forall t_1\dots t_{p(|x|)}\in\{0,1\}^{p(|x|)}\exists v\in\{0,1\}^{p(|x|)}\exists i\in \{1\dots p(|x|)\}|x\rangle\otimes |t_i\oplus v\rangle\in \qpromprobY{B}.
	\end{align}
	By definition of $\cclass{QMA}$, for any bitstring $y$ such that $|x\rangle\otimes |y\rangle\in \promprobY{B}$,
	\begin{align*}
	\exists |\psi\rangle \in \mathcal{H}_2^{s(|x|)}.\pr[Q_{x\circ y} \text{ accepts } |\psi\rangle] \ge 2/3.
	\end{align*}
	From Lemma \ref{lemma:eahousekeeping} we know we can collapse the classical $\exists$ quantifiers into the quantum one, obtaining  $\forall_c\exists_q$. This means that Eq. (\ref{eq:longAE}) is of the form required by a promise problem in $\text{cq-}\Pi_{2}$.
	
	Set $S'_x=\{y\in\{0,1\}^{p(|x|)}~:~|x\rangle\otimes |y\rangle\in \promprobN{B}\}$. For $x\in \promprobN{A}$, $|S'_x|\ge 1-2^{-r(|x|)}2^{p(|x|)}$, for any $r$ via amplification. Then by Corollary \ref{lemma:quantSimEA} we know that this can be written as a $\exists_c \forall_c$ statement about belonging to $\qpromprobN{B}$. By definition the membership condition for $\qpromprobN{B}$ is a $\forall_q$ statement. Again, the classical and quantum $\forall$ statements can be collapsed so we obtain a $\exists_c\forall_q$ statement for the NO instances, meaning that $\qpromprob{A}\in\text{cq-}\Pi_{2}$. 
\end{proof}
\section{Acknowledgements}
The authors thank Scott Aaronson, L\'aszl\'o Babai, Toby Cubitt, Aram Harrow, Will Matthews, Ashley Montanaro, Andrea Rocchetto, and Simone Severini for very helpful discussions. JL acknowledges financial support by the Engineering and Physical Sciences Research Council [grant number EP/L015242/1]. CEGG thanks UCL CSQ group for hospitality during the first semester of 2017, and acknowledges financial support from the Spanish MINECO (project MTM2014-54240-P) and MECD ``Jos\'e Castillejo''  program (CAS16/00339).


\begin{thebibliography}{9}
	\bibitem{ladner}  R. Ladner ``On the Structure of Polynomial Time Reducibility'', \emph{Journal of the ACM} 22(1): 155–171 (1975).
	
	\bibitem{bh} R. B. Boppana, J. Hastad ``Does co-NP have short interactive proofs?'', \emph{Information Processing Letters} 25(2) pp. 126-132 (1987).
	
	\bibitem{nauty} B.D. McKay, A. Piperno, ``Practical Graph Isomorphism, II'' \emph{Journal of Symbolic Computation} 60, pp. 94-112 (2014).
	
	\bibitem{an} D. Aharonov, T. Naveh, ``Quantum NP -- A Survey'', \emph{arXiv:quant-ph/0210077 (preprint)} (2002).
	
	\bibitem{babai2} L. Babai, ``Monte-Carlo algorithms in graph isomorphism testing'', \emph{Universit\'e de Montr\'eal Technical Report}, DMS:79-10 pp. 42 (1979).
	
	\bibitem{groupgraph} ``Group-Theoretic Algorithms and Graph Isomorphism'', \emph{Lecture Notes in Computer Science} 136, Editor: C. M. Hoffmann (1982).
	
	
	\bibitem{luks} E. M. Luks, ``Isomorphism of graphs of bounded valence can be tested in
	polynomial time'', \emph{Journal of Computer and System Science} 25(1) pp. 42–65 (1982).
	
	
	\bibitem{sims} C. C. Sims, ``Computation with permutation groups'', \emph{Proceedings of the Second ACM Symposium on Symbolic and Algebraic Manipulation} pp. 23-28 (1971).
	
	\bibitem{FHL} M. Furst, J. Hopcroft, E. Luks, ``Polynomial-time algorithms for permutation groups'', \emph{21st Annual Symposium on Foundations of Computer Science} (1980).
	
	\bibitem{luks2} E. M. Luks, ``Permutation groups and polynomial-time computation''. \emph{Groups and Computation, DIMACS Series in Discrete Mathematics and Theoretical Computer Science}
	11 pp. 139–175 (1993).

	
	
	\bibitem{swaptest} H. Buhrman, R. Cleve, J. Watrous,
	R. de Wolf. ``Quantum fingerprinting'', \emph{Physical
	Review Letters}, 87(16):167902 (2001).
	
	
	\bibitem{babai} L. Babai, ``Graph Isomorphism in Quasipolynomial Time'' \emph{arXiv:1512.03547 (preprint)} (2015).
	
	\bibitem{gknill} D. Gottesman, talk at International Conference on Group Theoretic Methods in Physics
	(1998), \emph{arXiv:quant-ph/9807006}.
	
	\bibitem{qip=pspace}
	R. Jain, Z. Ji, S. Upadhyay, J. Watrous, ``QIP=PSPACE'',
	\emph{arXiv:0907.4737 (preprint)} (2009).
	
	\bibitem{nc} M. Nielsen, I. Chuang, ``Quantum Computation and Quantum Information'', $10^{\text{th}}$ ed., \emph{Cambridge University Press} (2011).
	
	\bibitem{watrous} J. Watrous, ``Quantum Computational Complexity'' \emph{arXiv:0804.3401 (preprint)} (2008).
	
	\bibitem{aram} A. W. Harrow, ``The Church of the Symmetric Subspace'' \emph{arXiv:1308.6595 (preprint)} (2013).
	
	\bibitem{graphstates} M. Hein, W. D\"ur, J. Eisert, R. Raussendorf, M. Van den Nest, H.-J. Briegel ``Entanglement in Graph States and its Applications'', \emph{Proceedings of the International School of Physics ``Enrico Fermi'': Quantum Computers, Algorithms and Chaos} pp. 115-218 (2006).
	
	\bibitem{bqppolyhi} S. Aaronson, ``BQP and the Polynomial Hierarchy'', \emph{arXiv:0910.4698 (preprint)} (2009).
	
	\bibitem{septesting} G. Gutoski, P. Hayden, K. Milner, M. M. Wilde, 
	``Quantum interactive proofs and the complexity of separability testing''
	\emph{Theory of Computing}, 11(3), pp. 59-103 (2015).
	
	\bibitem{mont} A. Montanaro, ``Learning stabilizer states by Bell sampling'', \emph{arXiv:1707.04012 (preprint)} (2017).
	
	\bibitem{stabilizers} D. Gottesman, ``Stabilizer Codes and Quantum Error Correction'', \emph{arXiv:quant-ph/9705052 (PhD thesis)} (1997).
	
	\bibitem{ag} S. Aaronson, D. Gottesmann, ``Improved simulation of stabilizer circuits'' \emph{Physical Review A} 70, 052328 (2004).
	
	\bibitem{gmc} H. J. Garcia, I. L. Markov, A. W. Cross, ``Efficient Inner-product Algorithm for Stabilizer States'', \emph{arXiv:1210.6646 (preprint)} (2012). 
	
	\bibitem{yamakami} T Yamakami, ``Quantum NP and a Quantum Hierarchy'', 
	\emph{Proceedings of the 2nd IFIP International Conference on Theoretical Computer Science}, pp. 323-336 
	 (2002).
	
	\bibitem{gk} S. Gharibian, J. Kempe, ``Hardness of approximation for quantum problems'', \emph{Quantum Information and Computation} 14 (5 and 6) pp. 517-540 (2014).
	
	\bibitem{qszk} J. Watrous, ``Quantum statistical zero-knowledge'', \emph{	arXiv:quant-ph/0202111 (preprint)} (2002).
	
	\bibitem{advice} S. Aaronson, Quantum versus classical proofs and advice, \emph{arXiv:quant-ph/0604056 (preprint)} (2006).
	
\bibitem{gs} S. Goldwasser, M. Sipser, ``Private coins versus public coins in interactive proof systems'',
	\emph{Proceedings of the Eighteenth Annual ACM Symposium on Theory of Computing (STOC)}
	pp. 59-68 (1986).
	
	\bibitem{kobayashi} H. Kobayashi, F. Le Gall, H. Nishimura, ``Generalized Quantum Arthur-Merlin Games'', \emph{Proceedings of the 30th Conference on Computational Complexity (CCC2015)}, pp. 488-511 (2015).
	
	\bibitem{helstrom} C. W. Helstrom, ``Quantum Detection and Estimation Theory'', Academic Press, New York, (1976).
	
	\bibitem{twomessage} R. Jain, S. Upadhyay, J. Watrous, ``Two-message quantum interactive proofs are in PSPACE'' \emph{Proceedings of the 50th IEEE Conference on Foundations of Computer Science (FOCS), pp. 534–543} (2009).
	
	\bibitem{mw} C. Marriott, J. Watrous, ``Quantum Arthur-Merlin Games'', \emph{	arXiv:cs/0506068 (preprint)} (2005).
	
	\bibitem{hlm} A. Harrow, C. Y. Lin, A. Montanaro, ``Sequential measurements, disturbance and property testing'', \emph{arXiv:1607.03236 (preprint)} (2016).
	
	\bibitem{ab} S. Arora, B. Barak, ``Computational Complexity: A Modern Approach'' \emph{Cambridge University Press} (2009).
\end{thebibliography}
\end{document}